\documentclass[11pt]{article}
\usepackage{ifthen}
\usepackage{mdwlist}
\usepackage{amsmath,amssymb,amsfonts,amsthm,mathtools}
\usepackage{bm}
\usepackage[colorlinks=true,linkcolor=blue,citecolor=blue,urlcolor=blue]{hyperref}
\usepackage{enumitem}
\usepackage{graphicx}
\usepackage{xspace}
\usepackage{verbatim}
\usepackage{algpseudocode}
\usepackage[boxruled,vlined,linesnumbered]{algorithm2e}
	\SetKwFor{While}{While}{}{}
	\SetKwFor{For}{For}{}{}
	\SetKwIF{If}{ElseIf}{Else}{If}{}{elif}{Else}{}
	\SetKw{Return}{Return}
\usepackage[margin=1in]{geometry}
\usepackage{color}
\usepackage{thm-restate}
\usepackage{latexsym}
\usepackage{epsfig}
\usepackage{xcolor}
\usepackage{CJKutf8}

\def\eps{\ve}
\renewcommand{\epsilon}{\ve}
\def\ve{\varepsilon}

\newcommand{\E}{\mbox{\bf E}}

\newcommand{\pr}[2][]{\mathrm{Pr}\ifthenelse{\not\equal{}{#1}}{_{#1}}{}\!\left[#2\right]}

\providecommand{\poly}{\operatorname*{poly}}

\newtheorem{theorem}{Theorem}

\newtheorem{lemma}[theorem]{Lemma}

\newtheorem{corollary}[theorem]{Corollary}

\newtheorem{definition}[theorem]{Definition}

\numberwithin{theorem}{section}
\numberwithin{nontheorem}{section}
\numberwithin{proposition}{section}
\numberwithin{observation}{section}
\numberwithin{remark}{section}
\numberwithin{fact}{section}
\numberwithin{lemma}{section}
\numberwithin{claim}{section}
\numberwithin{corollary}{section}
\numberwithin{case}{section}
\numberwithin{dfn}{section}
\numberwithin{definition}{section}
\numberwithin{question}{section}
\numberwithin{openquestion}{section}
\numberwithin{res}{section}

\def \cA {{\cal A}}
\def \cC {{\cal C}}
\def \cD {{\cal D}}
\def \cL {{\cal L}}
\def \cS {{\cal S}}
\def \cV {{\cal V}}
\def \cW {{\cal W}}
\def \cX {{\cal X}}

\newcommand{\RR}{\mathbb{R}}

\newcommand{\ab}{k}
\newcommand{\ns}{n}
\newcommand{\dist}{\alpha}
\newcommand{\Zon}{Z^\ns}
\newcommand{\Xon}{X^\ns}
\newcommand{\Yon}{Y^\ns}
\newcommand{\dims}{p}

\newcommand{\posa}{u}
\newcommand{\posb}{v}

\newcommand{\absv}[1]{\left|#1\right|}
\newcommand{\norm}[1]{\left\lVert#1\right\rVert_1}
\newcommand{\norminf}[1]{\left\lVert#1\right\rVert_{\infty}}
\newcommand{\norminfone}[1]{\left\lVert#1\right\rVert_{\infty,1}}
\def \Paren#1{{\left({#1}\right)}}
 
\newcommand{\probof}[1]{\Pr\Paren{#1}}
\newcommand{\proboff}[2]{\Pr_{#2}\Paren{#1}}
\newcommand{\expectation}[1]{\mathbb{E}\left[#1\right]}
\newcommand{\expectationf}[2]{\mathbb{E}_{#2}\left[#1\right]}

\newcommand{\ham}[2]{d_{ham}(#1,#2)}

\def \Brack#1{{\left[{#1}\right]}}

\begin{document}

\title{Privately Learning Markov Random Fields}
\author{Huanyu Zhang\footnotemark[1] \and Gautam Kamath\footnotemark[2] \footnotemark[5] \and Janardhan Kulkarni\footnotemark[3] \footnotemark[5] \and Zhiwei Steven Wu\footnotemark[4] \footnotemark[5]}
\renewcommand{\thefootnote}{\fnsymbol{footnote}}
\footnotetext[1]{Cornell University. {\tt hz388@cornell.edu}. Supported by NSF \#1815893 and by NSF \#1704443. This work was partially done while the author was an intern at Microsoft Research Redmond.}
\footnotetext[2]{University of Waterloo. {\tt g@csail.mit.edu}. Supported by a University of Waterloo startup grant. Part of this work was done while supported as a Microsoft Research Fellow, as part of the Simons-Berkeley Research Fellowship program, and while visiting Microsoft Research Redmond.}
\footnotetext[3]{Microsoft Research Redmond. {\tt jakul@microsoft.com}.}
\footnotetext[4]{University of Minnesota. {\tt zsw@umn.edu}. Supported in part by the NSF FAI Award  \#1939606, a Google Faculty Research Award, a J.P. Morgan Faculty Award, a Facebook Research Award, and a Mozilla Research Grant.}
\footnotetext[5]{These authors are in alphabetical order.}
\renewcommand{\thefootnote}{\arabic{footnote}}

\maketitle 

\begin{abstract}
  We consider the problem of learning Markov Random Fields (including
  the prototypical example, the Ising model) under the constraint of
  differential privacy.  Our learning goals include both
  \emph{structure learning}, where we try to estimate the underlying
  graph structure of the model, as well as the harder goal of
  \emph{parameter learning}, in which we additionally estimate the
  parameter on each edge.  We provide algorithms and lower bounds for
  both problems under a variety of privacy constraints --
  namely pure, concentrated, and approximate differential privacy.
  While non-privately, both learning goals enjoy roughly the same
  complexity, we show that this is not the case under differential
  privacy.  In particular, only structure learning under approximate
  differential privacy maintains the non-private logarithmic
  dependence on the dimensionality of the data, while a change in
  either the learning goal or the privacy notion would necessitate a
  polynomial dependence. As a result, we show that the privacy
  constraint imposes a strong separation between these two learning
  problems in the high-dimensional data regime.
 \end{abstract}
 
\section{Introduction}
Graphical models are a common structure used to model high-dimensional
data, which find a myriad of applications in diverse research
disciplines, including probability theory, Markov Chain Monte Carlo,
computer vision, theoretical computer science, social network
analysis, game theory, and computational
biology~\cite{LevinPW09,Chatterjee05,Felsenstein04,DaskalakisMR11,GemanG86,Ellison93,MontanariS10}.
While statistical tasks involving general distributions over $\dims$
variables often run into the curse of dimensionality (i.e., an
exponential sample complexity in $\dims$), Markov Random Fields (MRFs)
are a particular family of undirected graphical models which are
parameterized by the ``order'' $t$ of their interactions.  Restricting
the order of interactions allows us to capture most distributions
which may naturally arise, and also avoids this severe dependence on
the dimension (i.e., we often pay an exponential dependence on $t$
instead of $\dims$).  An MRF is defined as follows, see
Section~\ref{sec:preliminaries} for more precise definitions and
notations we will use in this paper.
\begin{definition}
  Let $\ab, t, \dims \in \mathbb{N}$, $G = (V,E)$ be a graph on $\dims$ nodes, and $C_t(G)$ be the set of cliques of size at most $t$ in $G$.
  A \emph{Markov Random Field} with alphabet size $\ab$ and $t$-order interactions is a distribution $\mathcal{D}$ over $[\ab]^p$ such that
  \[
    \Pr_{X \sim \mathcal{D}}[X = x] \propto \exp\left(\sum_{I \in C_t(G)} \psi_I(x) \right),
  \]
  where $\psi_I : [\ab]^{p} \rightarrow \mathbb{R}$ depends only on varables in $I$.
\end{definition}

The case when $\ab = t = 2$ corresponds to the prototypical example of an MRF, the Ising model~\cite{Ising25} (Definition~\ref{def:ising}).
More generally, if $t = 2$, we call the model \emph{pairwise} (Definition~\ref{def:pairwise}), and if $\ab = 2$ but $t$ is unrestricted, we call the model a \emph{binary MRF} (Definition~\ref{def:mrf}). In this paper, we mainly look at these two special cases of MRFs.

Given the wide applicability of these graphical models, there has been
a great deal of work on the problem of graphical model
estimation~\cite{RavikumarWL10, SanthanamW12, Bresler15, VuffrayMLC16,
  KlivansM17, HamiltonKM17, RigolletH17, LokhovVMC18, WuSD19}.  That
is, given a dataset generated from a graphical model, can we infer
properties of the underlying distribution?  Most of the attention has
focused on two learning goals.

\begin{enumerate}
  \item \emph{Structure learning} (Definition~\ref{def:learn-struct}): Recover the set of non-zero edges in $G$.
  \item \emph{Parameter learning} (Definition~\ref{def:learn-params}): Recover the set of non-zero edges in $G$, as well as $\psi_I$ for all cliques $I$ of size at most $t$.
\end{enumerate}

It is clear that structure learning is easier than parameter learning.
Nonetheless, the sample complexity of both learning goals is known to be roughly equivalent.
That is, both can be performed using a number of samples which is only \emph{logarithmic} in the dimension $\dims$ (assuming a model of bounded ``width'' $\lambda$\footnote{This is a common parameterization of the problem, which roughly corresponds to the graph having bounded-degree, see Section~\ref{sec:preliminaries} for more details.}), thus facilitating estimation in very high-dimensional settings.

However, in modern settings of data analysis, we may be running our algorithms on datasets which are sensitive in nature.
For instance, graphical models are often used to model medical and genetic data~\cite{FriedmanLNP00,LagorAFH01} -- if our learning algorithm reveals too much information about individual datapoints used to train the model, this is tantamount to releasing medical records of individuals providing their data, thus violating their privacy.
In order to assuage these concerns, we consider the problem of learning graphical models under the constraint of \emph{differential privacy} (DP)~\cite{DworkMNS06}, considered by many to be the gold standard of data privacy.
Informally, an algorithm is said to be differentially private if its distribution over outputs is insensitive to the addition or removal of a single datapoint from the dataset (a more formal definition is provided in Section~\ref{sec:preliminaries}).
Differential privacy has enjoyed widespread adoption, including deployment in Apple~\cite{AppleDP17}, Google~\cite{ErlingssonPK14}, Microsoft~\cite{DingKY17}, and the US Census Bureau for the 2020 Census~\cite{DajaniLSKRMGDGKKLSSVA17}.

Our goal is to design algorithms which guarantee both: 

\begin{itemize}
  \item Accuracy: With probability greater than $2/3$, the algorithm learns the underlying graphical model;

  \item Privacy: The algorithm satisfies differential privacy, even when the dataset is not drawn from a graphical model. 
\end{itemize}

Thematically, we investigate the following question: how much additional data is needed to learn Markov Random Fields under the constraint of differential privacy?
As mentioned before, absent privacy constraints, the sample complexity is logarithmic in $\dims$.
Can we guarantee privacy with comparable amounts of data?
Or if more data is needed, how much more?

\subsection{Results and Techniques}
\label{sec:results-techniques}
We proceed to describe our results on privately learning Markov Random Fields.
In this section, we will assume familiarity with some of the most common notions of differential privacy: pure $\ve$-differential privacy, $\rho$-zero-concentrated differential privacy, and approximate $(\ve, \delta)$-differential privacy.
In particular, one should know that these are in (strictly) decreasing order of strength (i.e., an algorithm which satisfies pure DP gives more privacy to the dataset than concentrated DP), formal definitions appear in Section~\ref{sec:preliminaries}.
Furthermore, in order to be precise, some of our theorem statements will use notation which is defined later (Section~\ref{sec:preliminaries}) -- these may be skipped on a first reading, as our prose will not require this knowledge.

\paragraph{Upper Bounds.}
Our first upper bounds are for parameter learning.
First, we have the following theorem, which gives an upper bound for parameter learning pairwise graphical models under concentrated differential privacy, showing that this learning goal can be achieved with $O(\sqrt{\dims})$ samples.
In particular, this includes the special case of the Ising model, which corresponds to an alphabet size $k = 2$.
Note that this implies the same result if one relaxes the learning goal to structure learning, or the privacy notion to approximate DP, as these modifications only make the problem easier. Further details are given in Section~\ref{sec:ub-pair}.
\begin{theorem}
  \label{thm:est-ub}
  There exists an efficient $\rho$-zCDP algorithm which learns the parameters of a pairwise graphical model to accuracy $\dist$ with probability at least $2/3$, which requires a sample complexity of
  $$\ns = O\Paren{\frac{ \lambda^2 k^5 \log(\dims k) e^{O(\lambda)}}{\dist^4}+ \frac{\sqrt{\dims} \lambda^2 k^{5.5} \log^2(\dims k)e^{O(\lambda)}}{\sqrt{\rho} \alpha^3}}$$
\end{theorem}

This result can be seen as a private adaptation of the elegant work of~\cite{WuSD19} (which in turn builds on the structural results of~\cite{KlivansM17}).
Wu, Sanghavi, and Dimakis~\cite{WuSD19} show that $\ell_1$-constrained logistic regression suffices to learn the parameters of all pairwise graphical models. 
We first develop a private analog of this method, based on the private Franke-Wolfe method of Talwar, Thakurta, and Zhang~\cite{TalwarTZ14,TalwarTZ15}, which is of independent interest. This method is studied in Section~\ref{sec:PFW}.
\begin{theorem}
  \label{thm:log-reg}
If we consider the problem of private sparse logistic regression,
there exists an efficient $\rho$-zCDP algorithm that produces a parameter vector $w^{priv}$, such that with probability at least $1-\beta$, the empirical risk
\[ 
\cL(w^{priv}; D) - \cL(w^{erm}; D) = O\Paren{ \frac{\lambda^{\frac{4}{3}}\log(\frac{ \ns\dims}{\beta}) } {(\ns \sqrt{\rho})^{\frac{2}{3}}}}.
\]
\end{theorem}
We note that Theorem~\ref{thm:log-reg} avoids a polynomial dependence on the dimension $\dims$ in favor of a polynomial dependence on the ``sparsity'' parameter $\lambda$.
The greater dependence on $\dims$ which arises in Theorem~\ref{thm:est-ub} is from applying Theorem~\ref{thm:log-reg} and then using composition properties of concentrated DP.

We go on to generalize the results of~\cite{WuSD19}, showing that $\ell_1$-constrained logistic regression can also learn the parameters of binary $t$-wise MRFs.
This result is novel even in the non-private setting. Further details are presented in Section~\ref{sec:bin-mrf}.


The following theorem shows that we can learn the parameters of binary $t$-wise MRFs with $\tilde O(\sqrt{\dims})$ samples.
\begin{theorem}
Let $\cD$ be an unknown binary $t$-wise MRF with associated polynomial $h$. Then there exists an $\rho$-zCDP algorithm which, with probability at least $2/3$, learns the maximal monomials of $h$ to accuracy $\dist$, given $\ns$ i.i.d.\ samples $Z^1,\cdots, Z^{\ns} \sim \cD$, where
$$\ns =O\Paren{ \frac{ e^{5\lambda t} \sqrt{\dims} \log^2(\dims) }{\sqrt{\rho} \dist^{\frac{9}{2}}}+ \frac{  t \lambda^2 \sqrt{\dims} \log{\dims}}{\sqrt{\rho}\dist^2} +  \frac{e^{6\lambda t} \log(\dims)}{\dist^6} } . $$
\end{theorem}

To obtain the rate above, our algorithm uses the Private
  Multiplicative Weights (PMW) method by \cite{HardtR10} to estimate
  all parity queries of all orders no more than $t$. The PMW method
  runs in time exponential in $p$, since it maintains a distribution
  over the data domain. We can also obtain an \emph{oracle-efficient}
  algorithm that runs in polynomial time when given access to an
  empirical risk minimization oracle over the class of parities. By
  replacing PMW with such an oracle-efficient algorithm \textsc{sepFEM} in
  \cite{neworacle}, we obtain a slightly worse sample complexity
$$\ns =O\Paren{ \frac{ e^{5\lambda t} \sqrt{\dims} \log^2(\dims) }{\sqrt{\rho} \dist^{\frac{9}{2}}}+ \frac{  t \lambda^2 {\dims^{5/4}} \log{\dims}}{\sqrt{\rho}\dist^2} +  \frac{e^{6\lambda t} \log(\dims)}{\dist^6} } . $$

For the special case of structure learning under approximate differential privacy, we provide a significantly better algorithm.
In particular, we can achieve an $O(\log \dims)$ sample complexity, which improves exponentially on the above algorithm's sample complexity of $O(\sqrt{\dims})$.
The following is a representative theorem statement for pairwise graphical models, though we derive similar statements for binary MRFs of higher order.
\begin{theorem}
  \label{thm:struct-ub}
    There exists an efficient $(\varepsilon, \delta)$-differentially private algorithm which, with probability at least $2/3$, learns the structure of a pairwise graphical model, which requires a sample complexity of $$n = O\left(\frac{\lambda^2 \ab^4 \exp(14\lambda) \log(\dims \ab)\log(1/\delta)}{\varepsilon\eta^4}\right).$$
\end{theorem}
This result can be derived using stability properties of non-private algorithms.
In particular, in the non-private setting, the guarantees of algorithms for this problem recover the entire graph \emph{exactly} with constant probability.
This allows us to derive private algorithms at a multiplicative cost of $O(\log(1/\delta)/\varepsilon)$ samples, using either the propose-test-release framework~\cite{DworkL09} or stability-based histograms~\cite{KorolovaKMN09, BunNSV15}.
Further details are given in Section~\ref{sec:struct-ub}.

\paragraph{Lower Bounds.}

We note the significant gap between the aforementioned upper bounds: in particular, our more generally applicable upper bound (Theorem~\ref{thm:est-ub}) has a $O(\sqrt{\dims})$ dependence on the dimension, whereas the best known lower bound is $\Omega(\log \dims)$~\cite{SanthanamW12}.
However, we show that our upper bound is tight.
That is, even if we relax the privacy notion to approximate differential privacy, \emph{or} relax the learning goal to structure learning, the sample complexity is still $\Omega(\sqrt{\dims})$.
Perhaps surprisingly, if we perform both relaxations simultaneously, this falls into the purview of Theorem~\ref{thm:struct-ub}, and the sample complexity drops to $O(\log \dims)$.

First, we show that even under approximate differential privacy, learning the parameters of a graphical model requires $\Omega(\sqrt{\dims})$ samples. The formal statement is given in Section~\ref{sec:est-lb}.
\begin{theorem}[Informal]
  Any algorithm which satisfies approximate differential privacy and learns the parameters of a pairwise graphical model with probability at least $2/3$ requires $\poly(\dims)$ samples.
\end{theorem}
This result is proved by constructing a family of instances of binary pairwise graphical models (i.e., Ising models) which encode product distributions.
Specifically, we consider the set of graphs formed by a perfect matching with edges $(2i, 2i+1)$ for $i \in [\dims/2]$.
In order to estimate the parameter on every edge, one must estimate the correlation between each such pair of nodes, which can be shown to correspond to learning the mean of a particular product distribution in $\ell_\infty$-distance.
This problem is well-known to have a gap between the non-private and private sample complexities, due to methods derived from fingerprinting codes~\cite{BunUV14, DworkSSUV15, SteinkeU17a}, and differentially private Fano's inequality~\cite{AcharyaSZ20}.

Second, we show that learning the structure of a graphical model, under either pure or concentrated differential privacy, requires $\poly(\dims)$ samples. The formal theorem  appears in Section~\ref{sec:struct-lb}.
\begin{theorem}[Informal]
  Any algorithm which satisfies pure or concentrated differential privacy and learns the structure of a pairwise graphical model with probability at least $2/3$ requires $\poly(\dims)$ samples.
\end{theorem}
We derive this result via packing arguments~\cite{HardtT10,BeimelBKN14,AcharyaSZ20}, by showing that there exists a large number (exponential in $\dims$) of different binary pairwise graphical models which must be distinguished.
The construction of a packing of size $m$ implies lower bounds of $\Omega(\log m)$ and $\Omega(\sqrt{\log m})$ for learning under pure and concentrated differential privacy, respectively.

\subsubsection{Summary and Discussion}

We summarize our findings on privately learning Markov Random Fields in Table~\ref{tbl:ba-table}, focusing on the specific case of the Ising model. 
We note that qualitatively similar relationships between problems also hold for general pairwise models as well as higher-order binary Markov Random Fields.
Each cell denotes the sample complexity of a learning task, which is a combination of an objective and a privacy constraint.
Problems become harder as we go down (as the privacy requirement is tightened) and to the right (structure learning is easier than parameter learning).

The top row shows that both learning goals require only $\Theta(\log \dims)$ samples to perform absent privacy constraints, and are thus tractable even in very high-dimensional settings or when data is limited.
However, if we additionally wish to guarantee privacy, our results show that this logarithmic sample complexity is only achievable when one considers structure learning under approximate differential privacy.
If one changes the learning goal to parameter learning, \emph{or} tightens the privacy notion to concentrated differential privacy, then the sample complexity jumps to become polynomial in the dimension, in particular $\Omega(\sqrt{\dims})$.
Nonetheless, we provide algorithms which match this dependence, giving a tight $\Theta(\sqrt{\dims})$ bound on the sample complexity.


\begin{table*}[!htb]
\begin{center}
\begin{tabular}{| c | c | c |}
\hline
  \multicolumn{1}{|c|}{} & \multicolumn{1}{c|}{\textbf{Structure Learning}} & \multicolumn{1}{c|}{\textbf{Parameter Learning}} \\
\hline
  \textbf{Non-private}	& $\Theta(\log{\dims})$ (folklore) &  $\Theta(\log{\dims})$ (folklore) \\ \hline
  \textbf{Approximate DP}	& $\Theta(\log{\dims})$ (Theorems~\ref{thm:str-ub-pair})	&  $\Theta(\sqrt{\dims})$ (Theorems~\ref{thm:est-ub-ising} and~\ref{thm:est-lb})	\\ \hline
  \textbf{Zero-concentrated DP}	& $\Theta(\sqrt{\dims})$ (Theorems~\ref{thm:est-ub-ising} and~\ref{thm:str-ising})	&  $\Theta(\sqrt{\dims})$ (Theorems~\ref{thm:est-ub-ising} and~\ref{thm:est-lb})	\\ \hline
  \textbf{Pure DP}		&	$\Omega(\dims)$ (Theorem~\ref{thm:str-ising})&		$\Omega(\dims)$ (Theorem~\ref{thm:str-ising})	\\ \hline
\end{tabular}
\end{center}
  \caption{Sample complexity (dependence on $\dims$) of privately learning an Ising model.}
  \label{tbl:ba-table}
\end{table*}

\subsection{Related Work}
As mentioned before, there has been significant work in learning the structure and parameters of graphical models, see, e.g.,~\cite{ChowL68, CsiszarT06, AbbeelKN06, RavikumarWL10, JalaliJR11, JalaliRVS11, SanthanamW12, BreslerGS14b, Bresler15, VuffrayMLC16, KlivansM17, HamiltonKM17, RigolletH17, LokhovVMC18, WuSD19}.
Perhaps a turning point in this literature is the work of Bresler~\cite{Bresler15}, who showed for the first time that general Ising models of bounded degree can be learned in polynomial time.
Since this result, following works have focused on both generalizing these results to broader settings (including MRFs with higher-order interactions and non-binary alphabets) as well as simplifying existing arguments.
There has also been work on learning, testing, and inferring other statistical properties of graphical models~\cite{BhattacharyaM16, MartindelCampoCU16, DaskalakisDK17, MukherjeeMY18, Bhattacharya19}.
In particular, learning and testing Ising models in statistical distance have also been explored~\cite{DaskalakisDK18,GheissariLP18,DevroyeMR18a,DaskalakisDK19, BezakovaBCSV19}, and are interesting questions under the constraint of privacy.

Recent investigations at the intersection of graphical models and differential privacy include~\cite{BernsteinMSSHM17, ChowdhuryRJ19,McKennaSM19}.
Bernstein et al.~\cite{BernsteinMSSHM17} privately learn graphical models by adding noise to the sufficient statistics and use an expectation-maximization based approach to recover the parameters.
However, the focus is somewhat different, as they do not provide finite sample guarantees for the accuracy when performing parameter recovery, nor consider structure learning at all.
Chowdhury, Rekatsinas, and Jha~\cite{ChowdhuryRJ19} study differentially private learning of Bayesian Networks, another popular type of graphical model which is incomparable with Markov Random Fields.
McKenna, Sheldon, and Miklau~\cite{McKennaSM19} apply graphical models in place of full contingency tables to privately perform inference.

Graphical models can be seen as a natural extension of product distributions, which correspond to the case when the order of the MRF $t$ is  $1$.
There has been significant work in differentially private estimation of product distributions~\cite{BlumDMN05, BunUV14, DworkMNS06, SteinkeU17a, KamathLSU19, CaiWZ19, BunKSW19}.
Recently, this investigation has been broadened into differentially private distribution estimation, including sample-based estimation of properties and parameters, see, e.g.,~\cite{NissimRS07, Smith11, BunNSV15, DiakonikolasHS15, KarwaV18, AcharyaKSZ18, KamathLSU19, BunKSW19}.
For further coverage of differentially private statistics, see~\cite{KamathU20}.


\section{Preliminaries}
\label{sec:preliminaries}
Given an integer $n$, we let $[\ns] \coloneqq \{1,2,\cdots,\ns\}$. Given a set of points $X^1, \cdots, X^{\ns}$, we use superscripts, i.e., $X^i$ to denote the $i$-th datapoint. 
Given a vector $X \in \RR^{\dims}$, we use subscripts, i.e., $X_i$ to denote its $i$-th coordinate. We also use $X_{-i}$ to denote the vector after deleting the $i$-th coordinate, i.e. $X_{-i} = [X_1, \cdots, X_{i-1}, X_{i+1}, \cdots, X_{\dims}]$.

\subsection{Markov Random Field Preliminaries}
We first introduce the definition of the Ising model, which is a special case of general MRFs when $\ab=t=2$.
\begin{definition}
  \label{def:ising}
The $\dims$-variable Ising model is a distribution $\cD(A,\theta)$ on $\{-1, 1\}^{\dims}$ that satisfies
\begin{align}
 \proboff{Z = z }{} \propto \exp \Paren{\sum_{1\le i\le j \le \dims} A_{i,j} z_i z_j + \sum_{i \in [\dims]} \theta_i z_i}, \nonumber
\end{align}
where $A \in \RR^{\dims \times \dims}$ is a symmetric weight matrix with $A_{ii} = 0, \forall i \in [\dims]$ and $\theta \in \RR^{\dims}$ is a mean-field vector. 
The dependency graph of $ \cD(A,\theta)$ is an undirected graph $G= (V, E)$, with vertices $V = [\dims]$ and edges $E = \{ (i,j) : A_{i,j} \neq 0 \}$. The width of  $\cD(A,\theta)$ is defined as 
\begin{align}
\lambda (A,\theta)  = \max_{i\in[\dims]} \Paren{\sum_{j \in [\dims]} \absv{A_{i,j}} + \absv{\theta_i} }.\nonumber
\end{align}
Let $\eta (A,\theta)$ be the minimum edge weight in absolute value, i.e., $\eta (A,\theta) = \min_{i,j\in[\dims]: A_{i,j}\neq 0} \absv{A_{i,j}}.$
\end{definition}

We note that the Ising model is supported on $\{-1, 1\}^{\dims}$. A natural generalization is to generalize its support to $[\ab]^{\dims}$, and maintain pairwise correlations.

\begin{definition}
  \label{def:pairwise}
  The $\dims$-variable pairwise graphical model is a distribution $\cD(\cW,\Theta)$ on $[\ab]^{\dims}$ that satisfies
\begin{align}
  \proboff{Z = z }{} \propto \exp \Paren{\sum_{1\le i\le j \le \dims} W_{i,j}(z_i, z_j) + \sum_{i \in [\dims]} \theta_i(z_i)}, \nonumber
\end{align}
 where $\cW = \{ W_{i,j} \in \RR^{\ab \times \ab} : i \neq j \in [\dims]\}$ is a set of  weight matrices satisfying $W_{i,j} = W^T_{j,i}$, and $\Theta = \{\theta_i \in \RR^\ab : i \in [\dims]\}$ is a set of mean-field vectors.
The dependency graph of $ \cD(\cW,\Theta)$ is an undirected graph $G= (V, E)$, with vertices $V = [\dims]$ and edges $E = \{ (i,j) : W_{i,j} \neq 0 \}$. 
  The width of  $\cD(\cW,\Theta)$ is defined as 
\begin{align}
  \lambda (\cW,\Theta)  = \max_{i\in[\dims], a \in [\ab]} \Paren{\sum_{j \in [\dims] \backslash i} \max_{b \in [\ab]} \absv{W_{i,j}(a,b)} + \absv{\theta_i(a)} }.\nonumber
\end{align}
  Define $\eta(\cW, \Theta) = \min_{(i,j) \in E} \max_{a,b} |W_{i,j}(a,b)|$.
\end{definition}

Both the models above only consider pairwise interactions between
nodes.  In order to capture higher-order interactions, we examine the
more general model of Markov Random Fields (MRFs).  In this paper, we
will restrict our attention to MRFs over a binary alphabet (i.e.,
distributions over $\{\pm1\}^p$).  In order to define binary $t$-wise
MRFs, we first need the following definition of multilinear
polynomials, partial derivatives and maximal monomials.

\begin{definition}
Multilinear polynomial is defined as $h: \RR^{\dims} \rightarrow \RR$ such that $h(x)=\sum_{I} \bar{h}(I) \prod_{i \in I}x_i$ where $\bar{h}(I)$ denotes the coefficient of the monomial $\prod_{i \in I}x_i$ with respect to the variables $(x_i:i \in I)$. Let $\partial_i h(x) = \sum_{J: i \not\in J}\bar{h} (J \cup \{ i\})\prod_{j \in J}x_j$ denote the partial derivative of $h$ with respect to $x_i$. Similarly, for $I \subseteq [\dims]$, let $\partial_I h(x) = \sum_{J: J \cap I = \phi}\bar{h} (J \cup I )\prod_{j \in J}x_j$ denote the partial derivative of $h$ with respect to the variables $(x_i: i \in I)$.
We say $I \subseteq [\dims]$ is a maximal monomial of $h$ if $\bar{h}(J)=0$ for all $J \supset I$.
\end{definition}

Now we are able to formally define binary $t$-wise MRFs.

\begin{definition}
  \label{def:mrf}
For a graph $G= (V, E)$ on $\dims$ vertices, let $C_t(G)$ denotes all cliques of  size at most $t$ in G. A binary $t$-wise Markov random field on $G$ is a distribution $\cD$ on  $\{-1, 1\}^{\dims}$ which satisfies
\begin{align}
 \proboff{Z = z }{Z \sim \cD} \propto \exp \Paren{\sum_{I \in C_t(G) }\varphi_I(z) }, \nonumber
\end{align}
and each $\varphi_I:\RR^{\dims} \rightarrow \RR$ is a multilinear polynomial that depends only on the variables in $I$. 

We call $G$ the dependency graph of the MRF and $h(x) = \sum_{I \in C_t(G)} \varphi_I(x)$ the factorization polynomial of the MRF. The width of $\cD$ is defined as $\lambda = \max_{i \in [\dims]} \norm{\partial_i h}$, where $\norm{h} \coloneqq \sum_{I} \absv{\bar{h}(I)}$.
\end{definition}

Now we introduce the definition of $\delta$-unbiased distribution and its properties. The proof appears in~\cite{KlivansM17}.
\begin{definition}[$\delta$-unbiased]
Let $S$ be the alphabet set, e.g., $S=\{1, -1 \}$ for binary $t$-pairwise MRFs and $S=[\ab]$ for pairwise graphical models. A distribution $\cD$ on $S^{\dims}$ is $\delta$-unbiased if for $Z \sim \cD$, $\forall i \in [\dims]$, and any assignment $x \in S^{\dims-1}$ to $Z_{-i}$, $\min_{z \in S} \probof{Z_i = z |Z_{-i}= x } \ge \delta$.
\end{definition}

The marginal distribution of a $\delta$-unbiased distribution also satisfies $\delta$-unbiasedness.

\begin{lemma}
\label{lem:marginal}
Let $\cD$ be a $\delta$-unbiased on $S^{\dims}$, with alphabet set $S$. For $X \sim \cD$, $\forall i \in [\dims]$, the distribution of $X_{-i}$ is also $\delta$-unbiased.
\end{lemma}

The following lemmas provide $\delta$-unbiased guarantees for various graphical models.
\begin{lemma}
\label{lem:pairwise-unbiased}
Let $\cD(\cW,\Theta)$ be a pairwise graphical model  with alphabet size $\ab$ and width $\lambda(\cW,\Theta)$. Then $\cD(\cW,\Theta)$ is $\delta$-unbiased with $\delta=e^{-2\lambda(\cW,\Theta)} /\ab$. In particular, an Ising model $\cD(A,\theta)$ is $e^{-2\lambda(A,\theta)}/2$-unbiased.
\end{lemma}

\begin{lemma}
Let $\cD$ be a binary $t$-wise MRFs with width $\lambda$. Then $\cD$ is $\delta$-unbiased with $\delta=e^{-2\lambda}/2$. 
\end{lemma}

Finally, we define two possible goals for learning graphical models.
First, the easier goal is \emph{structure learning}, which involves recovering the set of non-zero edges.
\begin{definition}
\label{def:learn-struct}
  An algorithm learns the \emph{structure} of a graphical model if, given samples $Z_1, \dots, Z_n \sim \cD$, it outputs a graph $\hat G = (V,\hat E)$ over $V = [\dims]$ such that $\hat E = E$, the set of edges in the dependency graph of $\cD$. 
\end{definition}

The more difficult goal is \emph{parameter learning}, which requires the algorithm to learn not only the location of the edges, but also their parameter values.

\begin{definition}
\label{def:learn-params}
  An algorithm learns the \emph{parameters} of an Ising model (resp.\ pairwise graphical model) if, given samples $Z_1, \dots, Z_n \sim \cD$, it outputs a matrix $\hat A$ (resp.\ set of matrices $\hat \cW$) such that $\max_{i,j \in [\dims]} |A_{i,j} -\hat A_{i,j}| \leq \alpha$ (resp.\ $|W_{i,j}(a,b) - \widehat W_{i,j}(a,b)| \leq \alpha$, $\forall i\neq j \in [\dims], \forall a,b \in [\ab]$).
\end{definition}

\begin{definition}
An algorithm learns the \emph{parameters} of a binary $t$-wise MRF with associated polynomial $h$ if, given samples $X^1, \dots, X^n \sim \cD$, it outputs another multilinear polynomial $u$ such that
that for all maximal monomial $I \subseteq [\dims]$, $\absv{\bar h(I) -\bar u(I)} \le \dist$.
\end{definition}

\subsection{Privacy Preliminaries}
A \emph{dataset} $X = (X^1,\dots,X^n) \in \cX^n$ is a collection of points from some universe $\cX$. We say that two datasets $X$ and $X^{\prime}$ are neighboring, which are denoted as $X \sim X^{\prime}$ if they differ in exactly one single point. In our work we consider a few different variants of differential privacy.  The first is the standard notion of differential privacy.

\begin{definition}[Differential Privacy (DP)~\cite{DworkMNS06}] 
A randomized algorithm $\cA: \cX^n \rightarrow \cS$ satisfies \emph{$(\eps,\delta)$-differential privacy ($(\eps,\delta)$-DP)} if for every pair of neighboring datasets $X, X' \in \cX^n$, and any event $S \subseteq \cS$, 
\[
\probof {\cA(X) \in S} \leq e^{\eps} \probof{\cA(X^{\prime}) \in S} + \delta.
\]
\end{definition}

The second is \emph{concentrated differential privacy}~\cite{DworkR16}. In this work, we specifically consider its refinement \emph{zero-mean concentrated differential privacy}~\cite{BunS16}.
\begin{definition}[Concentrated Differential Privacy (zCDP)~\cite{BunS16}]
    A randomized algorithm $\cA: \cX^n \rightarrow \cS$
    satisfies \emph{$\rho$-zCDP} if for
    every pair of neighboring datasets $X, X' \in \cX^n$,
    $$\forall \alpha \in (1,\infty)~~~D_\alpha\left(M(X)||M(X')\right) \leq \rho\alpha,$$
    where $D_\alpha\left(M(X)||M(X')\right)$ is the
    $\alpha$-R\'enyi divergence between $M(X)$ and $M(X')$.
\end{definition}


\noindent The following lemma quantifies the relationships between $(\eps,0)$-DP, $\rho$-zCDP and $(\eps,\delta)$-DP.
\begin{lemma}[Relationships Between Variants of DP~\cite{BunS16}] 
\label{lem:dpdefns}
For every $\eps \geq 0$,
\begin{enumerate}
\item If $\cA$ satisfies $(\eps,0)$-DP, then $\cA$ is $\frac{\eps^2}{2}$-zCDP.
\item If $\cA$ satisfies $\frac{\eps^2}{2}$-zCDP, then $\cA$ satisfies $(\frac{\eps^2}{2} + \eps \sqrt{2 \log(\frac{1}{\delta})},\delta)$-DP for every $\delta > 0$.
\end{enumerate}
\end{lemma}
Roughly speaking, pure DP is stronger than zero-concentrated DP, which is stronger than approximate DP.

A crucial property of all the variants of differential privacy is that they can be composed adaptively.  By adaptive composition, we mean a sequence of algorithms $\cA_1(X),\dots,\cA_T(X)$ where the algorithm $\cA_t(X)$ may also depend on the outcomes of the algorithms $\cA_1(X),\dots,\cA_{t-1}(X)$.
\begin{lemma}[Composition of zero-concentrated DP~\cite{BunS16}] \label{lem:dpcomp} If $\cA$ is an adaptive composition of
  differentially private algorithms $\cA_1,\dots, \cA_T$, and $\cA_1,\dots,\cA_T$ are
  $\rho_1,\dots,\rho_T$-zCDP respectively, then $\cA$ is $\rho$-zCDP for
  $\rho = \sum_t \rho_t$.  
\end{lemma}


\section{Parameter Learning of Pairwise Graphical Models}

\subsection{Private Sparse Logistic Regression}
\label{sec:PFW}

As a subroutine of our parameter learning algorithm, we consider the
following problem: given a training data set $D$ consisting of n pairs
of data $D = \{d^j\}_{j=1}^{\ns}= \{ (x^j, y^j)\}_{j=1}^{\ns}$, where
$x^j \in \RR^p$ and $y^j \in \RR$, a constraint set $\cC \in \RR^p$,
and a loss function $\ell: \cC \times \RR^{p+1} \rightarrow \RR$, we
want to find
$w^{erm} = \arg\min_{w\in \cC} ~\cL(w;D) = \arg\min_{w\in \cC} ~
\frac1\ns{\sum_{j=1}^{\ns} \ell(w; d^j)}$ with a zCDP constraint. This
problem was previously studied in~\cite{TalwarTZ14}. Before stating
their results, we need the following two definitions.  The first
definition is regarding Lipschitz continuity.
\begin{definition}
A function $\ell : \cC \rightarrow \RR$ is $L_1$-Lipschitz with respect to $\ell_1$ norm, if the following holds.
$$\forall w_1, w_2 \in \cC, \absv{\ell(w_1) - \ell(w_2)} \le L_1 \norm{w_1-w_2}.$$
\end{definition}

The performance of the algorithm also depends on the ``curvature" of the loss function, which is defined below, based on the definition of~\cite{Clarkson10, Jaggi13}. 
A side remark is that this is a strictly weaker constraint than smoothness~\cite{TalwarTZ14}.

\begin{definition}[Curvature constant]
For $\ell:\cC \rightarrow \RR$, $\Gamma_{\ell}$ is defined as 
$$ \Gamma_{\ell} = \sup_{w_1, w_2 \in \cC, \gamma \in (0,1], w_3 = w_1 + \gamma (w_2-w_1)}\frac{2}{\gamma^2} \Paren{\ell(w_3) - \ell(w_1) - \langle w_3-w_1, \nabla \ell(w_1) \rangle}.$$
\end{definition}

Now we are able to introduce the algorithm and its theoretical guarantees.

\begin{algorithm}[H]

\KwIn{Data set: $D = \{d^1,\cdots,d^{\ns}\}$, loss function: $\cL(w;D) = \frac1{\ns} \sum_{j=1}^{\ns} \ell(w;d^j)$ (with Lipschitz constant $L_1$), privacy parameters: $\rho$, convex set: $\cC = conv(S)$ with $\norm{\cC} \coloneqq  \max_{s \in S} \norm{s}$, iteration times: $T$}

Initialize $w$ from an arbitrary point in $\cC$

\For{$t=1$ to $T-1$}{
$\forall s\in S$, $\alpha_s \leftarrow \langle s, \nabla \cL(w;D)\rangle+ \text{Lap} \Paren{ 0, \frac{L_1 \norm{C} \sqrt{T}}{\ns \sqrt{\rho}}}$

$\tilde{w_t} \leftarrow \arg\min_{s \in S} \alpha_s$

$w_{t+1} \leftarrow (1-\mu_t)w_{t} + \mu_t \tilde{w_t}$, where $\mu_t = \frac{2}{t+2}$
}
\KwOut{ $w^{priv} = w_T$}

\caption{$\cA_{PFW}( D, \cL, \rho, \cC):$ Private Frank-Wolfe Algorithm}
\label{alg:privateFW}
\end{algorithm}

\begin{lemma}[Theorem 5.5 from~\cite{TalwarTZ14}]
\label{lem:PrivateFW}
Algorithm~\ref{alg:privateFW} satisfies $\rho$-zCDP. Furthermore, let $L_1$, $\norm{\cC}$ be defined as in Algorithm~\ref{alg:privateFW}. Let $\Gamma_{\ell}$ be an upper bound on the curvature constant for the loss function $\ell(\cdot;d)$ for all $d$ and $\absv{S}$ be the number of extreme points in $S$.
If we set $T=\frac{\Gamma_{\ell}^{\frac{2}{3}} (n \sqrt{\rho})^{\frac{2}{3}}}{L_1\norm{\cC}^{\frac{2}{3}}}$, then with probability at least $1-\beta$ over the randomness of the algorithm,

$$ \cL(w^{priv}; D) - \cL(w^{erm}; D) = O\Paren{ \frac{\Gamma_{\ell}^{\frac13} (L_1\norm{\cC}) ^{\frac{2}{3}}\log(\frac{ \ns\absv{S}}{\beta})} {(\ns \sqrt{\rho})^{\frac{2}{3}}}}.$$

\end{lemma}
\begin{proof}
The utility guarantee is proved in~\cite{TalwarTZ14}. Therefore, it is enough to prove the algorithm satisfies $\rho$-zCDP.
According to the definition of the Laplace mechanism, every iteration of the algorithm satisfies $(\sqrt{\frac{\rho}{T}},0)$-DP, which naturally satisfies $\frac{\rho}{T}$-zCDP by Lemma~\ref{lem:dpdefns}. Then, by the composition theorem of zCDP (Lemma~\ref{lem:dpcomp}), the algorithm satisfies $\rho$-zCDP.
\end{proof}

If we consider the specific problem of sparse logistic regression, we will get the following corollary.
\begin{corollary}
\label{cor:EmpiricalError}

If we consider the problem of sparse logistic regression, i.e., $\cL(w;D) = \frac1{n} \sum_{j=1}^n \log(1+e^{-y^j \langle w, x^j \rangle})$, with the constraint that $\cC = \{w: \norm{w} \le \lambda\}$, and we further assume that $\forall j, \norminf{x^j} \le 1, y^j \in \{\pm 1\}$, let $T =  \lambda^{\frac{2}{3}} (\ns \sqrt{\rho})^{\frac{2}{3}}$,  then with probability at least $1-\beta$ over the randomness of the algorithm,
$$ \cL(w^{priv}; D) - \cL(w^{erm}; D) = O\Paren{ \frac{\lambda^{\frac{4}{3}}\log(\frac{ \ns\dims}{\beta}) } {(\ns \sqrt{\rho})^{\frac{2}{3}}}}.$$

Furthermore, the time complexity of the algorithm is $O(T \cdot \Paren{\ns \dims+\dims^2}) =O\Paren{ \ns^{\frac{2}{3}}\cdot \Paren{\ns\dims + \dims^2}}$.
\end{corollary}

\begin{proof}
First let we show $L_1 \le 2$. If we fix sample $d = (x,y)$, then for any $w_1, w_2\in \cC$,
$$\absv{ \ell(w_1;d) -\ell(w_2;d)} \le \max_w \norminf{\nabla_w \Paren{ \ell(w;d)}} \cdot \norm{w_1 - w_2}.$$
Since $\nabla_w \Paren{ \ell(w;d)} = \Paren{\sigma(\langle w, x \rangle) - y} \cdot x$, we have $\norminf{\nabla_w \Paren{ \ell(w;d)}} \le 2$.

  Next, we wish to show $\Gamma_{\ell} \le \lambda^2$.
  We use the following lemma from~\cite{TalwarTZ14}.
\begin{lemma}[Remark 4 in~\cite{TalwarTZ14}]
For any $q,r \ge 1$ such that $\frac1q+\frac1r=1$, $\Gamma_{\ell}$ is upper bounded by $\alpha {\left\lVert \cC \right\rVert_q}^2$, where $\alpha = \max_{w \in \cC, {\left\lVert v \right\rVert_q=1}}  {\left\lVert \nabla^2 \ell(w) \cdot v \right\rVert_q}$.
\end{lemma}

If we take $q=1, r = +\infty$, then $\Gamma_{\ell} \le \alpha \lambda^2$, where 
$$\alpha = \max_{w \in \cC, \norm{v}=1}   \norminf{ \nabla^2 \ell(w;d) \cdot v} \le \max_{i,j \in[\dims]} \Paren{ \nabla^2 \ell(w;d)}_{i,j}.$$
We have $\alpha \le 1$, since $\nabla^2 \ell(w;d) = \sigma(\langle w, x \rangle)  \Paren{1- \sigma(\langle w, x \rangle)} \cdot x x^T$, and $\norminf{x}\le 1$, 

Finally given $\cC = \{w: \norm{w}\le 1 \}$, the number of extreme points of $S$ equals $2\dims$. By replacing all these parameters in Lemma~\ref{lem:PrivateFW}, we have proved the loss guarantee in the corollary. 

With respect to the time complexity, we note that the time complexity of each iteration is $O\Paren{\ns\dims+\dims^2}$ and there are $T$ iterations in total.
\end{proof}

Now if we further assume the data set $D$ is drawn i.i.d.\ from some underlying distribution $P$, the following lemma from learning theory relates the true risk and the empirical risk, which shall be heavily used in the following sections.

\begin{theorem}
\label{thm:generalization_error}
If we consider the same problem setting and assumptions as in Corollary~\ref{cor:EmpiricalError}, and we further assume that the training data set $D$ is drawn i.i.d.\ from some unknown distribution $P$, then with probability at least $1-\beta$ over the randomness of the algorithm and the training data set,
$$ \expectationf{\ell(w^{priv};(X,Y))}{(X,Y)\sim P} - \expectationf{\ell(w^*;(X,Y))}{(X,Y)\sim P}  = O\Paren{ \frac{\lambda^{\frac{4}{3}}\log(\frac{ \ns\dims}{\beta}) } {(\ns \sqrt{\rho})^{\frac{2}{3}}} + \frac{\lambda \log\Paren{\frac1{\beta}}}{\sqrt{n}} },$$
where $w^* = \arg\min_{w\in C}\expectationf{\ell(w;(X,Y))}{(X,Y)\sim P} $.
\end{theorem}

\begin{proof}
 By triangle inequality,
\begin{align*}
&~\expectationf{\ell(w^{priv};(X,Y))}{(X,Y)\sim P} - \expectationf{\ell(w^*;(X,Y))}{(X,Y)\sim P} \nonumber\\
\le& \absv{ \expectationf{\ell(w^{priv};(X,Y))}{(X,Y)\sim P} - \frac1{\ns} \sum_{m=1}^{\ns} \ell(w^{priv};d^m) } + \absv{  \frac1{\ns}\sum_{m=1}^{\ns} \ell(w^{priv};d^m) - \frac1{\ns}\sum_{m=1}^{\ns} \ell(w^{erm};d^m) } \nonumber\\
+& \Paren{  \frac1{\ns}\sum_{m=1}^{\ns} \ell(w^{erm};d^m) - \frac1{\ns}\sum_{m=1}^{\ns} \ell(w^*;d^m) } + \absv{ \expectationf{\ell(w^{*};(X,Y))}{(X,Y)\sim P} - \frac1{\ns}\sum_{m=1}^{\ns} \ell(w^{*};d^m) } \nonumber
\end{align*}

Now we need to bound each term. We firstly bound the first and last term simultaneously. By the generalization error bound (Lemma 7 from~\cite{WuSD19}), they are bounded by $O\Paren{ \frac{\lambda \log\Paren{\frac1{\beta}}}{\sqrt{n}}}$ simultaneously, with probability greater than $1 - \frac{2}{3}\beta$. Then we turn to the second term, by Corollary~\ref{cor:EmpiricalError}, with probability greater than $1 - \frac{1}{3}\beta $, it is bounded by $ O\Paren{ \frac{\lambda^{\frac{4}{3}}\log(\frac{ \ns\dims}{\beta}) } {(\ns \sqrt{\rho})^{\frac{2}{3}}}}$. Finally we bound the third term. According to the definition of $w^{erm}$, the third term should be smaller than 0.
Therefore, by union bound, $ \expectationf{\ell(w^{priv};(X,Y))}{(X,Y)\sim P} - \expectationf{\ell(w^*;(X,Y))}{(X,Y)\sim P}  = O\Paren{ \frac{\lambda^{\frac{4}{3}}\log(\frac{ \ns\dims}{\beta}) } {(\ns \sqrt{\rho})^{\frac{2}{3}}} + \frac{\lambda \log\Paren{\frac1{\beta}}}{\sqrt{n}} }$, with probability greater than $1-\beta$.
\end{proof}


\subsection{Privately Learning Ising Models}

We first consider the problem of estimating the weight matrix of the Ising model. 
To be precise, given $\ns$ i.i.d.\ samples $ \{z^1, \cdots, z^{\ns}\}$ generated from an unknown distribution $\cD(A,\theta)$, our goal is to design an $\rho$-zCDP estimator $\hat{A}$ such that with probability at least $\frac{2}{3}$, $\max_{i,j\in[\dims]}\absv{A_{i,j} - \hat{A}_{i,j}} \le \dist$.

An observation of the Ising model is that for any node $Z_i$, the probability of $Z_i=1$ conditioned on the values of the remaining nodes $Z_{-i}$ follows from a sigmoid function. The next lemma comes from~\cite{KlivansM17}, which formalizes this observation.

\begin{lemma}
\label{lem:IsingToLR}
Let $Z \sim \cD(A,\theta)$ and $Z \in \{-1,1\}^{\dims}$, then $\forall i \in [\dims]$, $\forall x \in \{-1,1\}^{[\dims] \backslash \{i\}}$,
\begin{align}
\probof{Z_i=1|Z_{-i}=x} &= \sigma\Paren{ \sum_{j\neq i}2A_{i,j}x_j + 2\theta_i} = \sigma\Paren{\langle w, x^{\prime} \rangle}.\nonumber
\end{align}
where $w= 2[A_{i,1},\cdots, A_{i,i-1}, A_{i,i+1},\cdots, A_{i,\dims},\theta_i]  \in \RR^{\dims}$, and $x^{\prime} = [x,1] \in \RR^{\dims}$.
\end{lemma}

\begin{proof}
The proof is from~\cite{KlivansM17}, and we include it here for completeness. 
 According to the definition of the Ising model,
\begin{align}
\probof{Z_i=1|Z_{-i}=x} &= \frac{\exp \Paren{\sum\limits_{j\neq i}A_{i,j}x_j +\sum\limits_{j\neq i} \theta_j + \theta_i} }{\exp \Paren{\sum\limits_{j\neq i} A_{i,j}x_j +\sum_{j\neq i} \theta_j + \theta_i}+\exp\Paren{\sum_{j\neq i}-A_{i,j}x_j +\sum_{j\neq i} \theta_j - \theta_i}}\nonumber\\
& =\sigma\Paren{ \sum_{j\neq i}2A_{i,j}x_j + 2\theta_i}.\nonumber
\end{align}
\end{proof}

By Lemma~\ref{lem:IsingToLR}, we can estimate the weight matrix by
solving a logistic regression for each node, which is utilized
in~\cite{WuSD19} to design non-private estimators. Our algorithm uses
the private Frank-Wolfe method to solve the per-node logistic regression
problem, achieving the following theoretical guarantee.

\begin{algorithm}[H]

\KwIn{ $\ns$ samples $\{ z^1,\cdots, z^{\ns}\}$, where $z^m \in \{\pm1\}^{\dims}$ for $m\in [\ns]$; an upper bound on $\lambda(A,\theta) \le \lambda$, privacy parameter $\rho$}

\For{$i=1$ to $\dims$}{

$\forall m \in [\ns]$, $x^m \leftarrow [z_{-i}^{m},1]$, $y^m \leftarrow z_i^m$

$w^{priv} \leftarrow \cA_{PFW}( D, \cL, \rho^{\prime}, \cC)$, where $\rho^{\prime} = \frac{\rho}{\dims}$,   $D=\{\Paren{x^m, y^m} \}_{m=1}^{\ns}$, $\cL(w;D) =  \frac1{\ns}\sum_{m=1}^{\ns}\log\Paren{1+e^{-y^m \langle w, x^m \rangle}}$, $\cC = \{ \norm{w}\le 2\lambda\}$

$\forall j \in \dims$, $\hat{A}_{i,j} \leftarrow \frac1{2} w^{priv}_{\tilde{j}}$, where $\tilde{j} =j $ when $j<i$ and $\widetilde{j}  = j-1$ if $j>i$

}

\KwOut{$\hat{A} \in \RR^{\dims \times \dims}$}

\caption{Privately Learning Ising Models}
\label{alg:DPIsing}
\end{algorithm}

\begin{theorem}
\label{thm:est-ub-ising}
Let $\cD(A,\theta)$ be an unknown $\dims$-variable Ising model with $\lambda(A,\theta)\le \lambda$.
There exists an efficient $\rho$-zCDP algorithm which outputs a weight matrix $\hat{A} \in \RR^{\dims \times \dims}$ such that 
with probability greater than $2/3$, $\max_{i,j \in [\dims]} \absv{A_{i,j} - \hat{A}_{i,j}} \le \dist$ if the number of i.i.d.\ samples satisfies
$$\ns = \Omega\Paren{\frac{ \lambda^2 \log(\dims) e^{12\lambda}}{\dist^4}+ \frac{\sqrt{\dims} \lambda^2 \log^2(\dims)e^{9\lambda}}{\sqrt{\rho} \alpha^3}}.$$
\end{theorem}

\begin{proof}
We first prove that Algorithm~\ref{alg:DPIsing}  satisfies $\rho$-zCDP. Notice that in each iteration, the algorithm solves a private sparse logistic regression under $\frac{\rho}{\dims}$-zCDP. Therefore, Algorithm~\ref{alg:DPIsing} satisfies $\rho$-zCDP  by composition (Lemma~\ref{lem:dpcomp}).

For the accuracy analysis, we start by looking at the first iteration  ($i=1$) and showing that $\absv{A_{1,j} - \hat{A}_{1,j}} \le \dist$, $\forall j \in [\dims]$, with probability greater than $1-\frac1{10\dims}$. 

Given a random sample $Z \sim \cD(A,\theta)$, we let  $X = [Z_{-1},1]$, $Y = Z_1$. From Lemma~\ref{lem:IsingToLR}, $\probof{Y=1|X=x} = \sigma\Paren{\langle w^*, x\rangle}$, where $w^* = 2[A_{1,2},\cdots,A_{1,\dims},\theta_1]$. We also note that $\norm{w^*}\le 2\lambda$, as a consequence of the width constraint of the Ising model.

For any $\ns$ i.i.d.\ samples $\{ z^{m}\}_{m=1}^{\ns}$ drawn from the Ising model, let $x^m = [z_{-1}^{m},1]$ and $y^m = z_1^m$, it is easy to check that each $(x^m,y^m)$ is the realization of $(X,Y)$. Let $w^{priv}$ be the output of  $\cA\Paren{ D, \cL, \frac{\rho}{\dims},  \{w: \norm{w}\le 2\lambda \}}$, where $D = \{(x^m,y^m)\}_{m=1}^{\ns}$.
By Lemma~\ref{thm:generalization_error}, when $\ns =O\Paren{ \frac{ \sqrt{\dims} \lambda^2 \log^2(\dims) }{\sqrt{\rho} \gamma^{\frac{3}{2}}}+ \frac{ \lambda^2 \log(\dims)}{\gamma^2} }$, with probability greater than $1-\frac{1}{10\dims}$, $ \expectationf{\ell(w^{priv};(X,Y))}{Z\sim \cD(A,\theta)} - \expectationf{\ell(w^*;(X,Y))}{Z\sim \cD(A,\theta)} \le \gamma.$

We will use  the following lemma from~\cite{WuSD19}. Roughly speaking, with the assumption that the samples are generated from an Ising model, any estimator $w^{priv}$ which achieves a small error in the loss $\cL$  guarantees an accurate parameter recovery in $\ell_{\infty}$ distance.

\begin{lemma}
\label{lem:parameter-error-ising}
Let $P$ be a distribution on $\{-1,1\}^{\dims-1} \times \{-1,1\}$. Given $u_1\in \RR^{\dims-1}, \theta_1\in \RR$, suppose $\probof{Y=1|X=x} = \sigma\Paren{\langle u_1,x \rangle+\theta_1}$ for $(X,Y) \sim P$. 
If the marginal distribution of $P$ on $X$ is $\delta$-unbiased, 
and $\expectationf{\log\Paren{1+e^{-Y \Paren{ \langle u_1, X \rangle+\theta_1}}} } {(X,Y)\sim P} - \expectationf{ \log\Paren{1+e^{-Y \Paren{ \langle u_2, X \rangle+\theta_2}}} }{(X,Y)\sim P} \le \gamma$ for some $u_2\in \RR^{\dims-1}, \theta_2\in \RR$, and $\gamma \le \delta e^{-2\norm{u_1} -2\norm{\theta_1}-6}$, then $\norminf{u_1-u_2} = O(e^{\norm{u_1} +\norm{\theta_1}} \cdot \sqrt{\gamma/\delta}).$
\end{lemma}

By Lemma~\ref{lem:marginal}, Lemma~\ref{lem:pairwise-unbiased} and Lemma~\ref{lem:parameter-error-ising}, if $ \expectationf{\ell(w^{priv};(X,Y))}{Z\sim \cD(A,\theta)} - \expectationf{\ell(w^*;(X,Y))}{Z\sim \cD(A,\theta)} \le O\Paren{\alpha^2 e^{-6\lambda}}$, we have $\norminf{w^{priv} - w^*} \le \alpha$. By replacing $\gamma = \alpha^2 e^{-6\lambda}$, we prove that $\norminf{A_{1,j} - \hat{A}_{1,j}} \le \dist$ with probability greater than $1-\frac1{10\dims}$. Noting that similar argument works for the other iterations and non-overlapping part of the matrix is recovered in different iterations. By union bound over $\dims$ iterations, we prove that with probability at least $\frac{2}{3}$, $\max_{i,j\in[\dims]}\absv{A_{i,j} - \hat{A}_{i,j}} \le \dist$.

Finally, we note that the time compexity of the algorithm is $poly(\ns, \dims)$ since the private Frank-Wolfe algorithm is time efficient by Corollary~\ref{cor:EmpiricalError}.
\end{proof}

\subsection{Privately Learning Pairwise Graphical Models}
\label{sec:ub-pair}

Next, we study parameter learning for pairwise graphical models over general alphabet. Given $\ns$ i.i.d.\ samples $ \{z^1, \cdots, z^{\ns}\}$ drawn from an unknown distribution $\cD(\cW,\Theta)$, we want to design an $\rho$-zCDP estimator $\hat{\cW}$ such that with probability at least $\frac{2}{3}$, $\forall i \neq j \in [\dims], \forall \posa, \posb \in [\ab], \absv{W_{i,j}(\posa, \posb) - \widehat{W}_{i,j}(\posa, \posb)} \le \dist$. To facilitate our presentation, we assume that $\forall i \neq j \in [\dims]$, every row (and column) vector of $W_{i,j}$ has zero mean.\footnote{The assumption that $W_{i,j}$ is centered is without loss of generality and widely used in the literature~\cite{KlivansM17, WuSD19}. We present the argument here for completeness. Suppose the $a$-th row of $W_{i,j}$ is not centered, i.e., $\sum_{b} W_{i,j} (a,b) \neq 0$, we can define $W^{\prime}_{i,j} (a,b) = W_{i,j} (a,b)  - \frac1k \sum_{b} W_{i,j} (a,b)$ and $\theta^{\prime}_{i}(a)  = \theta_i(a)+\frac1k \sum_{b} W_{i,j} (a,b)$, and the probability distribution remains unchanged.}

Analogous to Lemma~\ref{lem:IsingToLR} for the Ising model, a pairwise graphical model has the following property, which can be utilized to recover its parameters. 
\begin{lemma}[Fact 2 of~\cite{WuSD19}]
\label{lem:PairwiseToLR}
Let $Z \sim \cD(\cW,\Theta)$ and $Z \in [\ab]^\dims$. For any $i \in [\dims]$, any $\posa \neq \posb \in [\ab]$, and any $x \in [\ab]^{\dims-1}$,

$$\probof{Z_i  =\posa | Z_i \in \{\posa, \posb\}, Z_{-i}=x }  = \sigma \Paren{\sum_{j\neq i}\Paren{W_{i,j}(\posa, x_j) -W_{i,j}(\posb, x_j) }+\theta_i(\posa)-\theta_i(\posb)}.
$$
\end{lemma}

Now we introduce our algorithm. Without loss of generality, we consider estimating $W_{1,j}$ for all $j \in [\dims]$ as a running example. We fix a pair of values $(\posa, \posb)$, where $\posa,\posb \in [\ab]$ and $\posa \neq \posb$. Let $S_{\posa,\posb}$ be the samples where $Z_1  \in \{\posa,\posb \}$. In order to utilize Lemma~\ref{lem:PairwiseToLR}, we perform the following transformation on the samples in $S_{\posa,\posb}$: for the $m$-th sample $z^m$, let $y^m =1$ if $z_1^m = \posa$, else $y^m = -1$. And $x^m$ is the one-hot encoding of the vector $[z^m_{-1},1]$, where $\text{OneHotEncode}(s)$ is a mapping from $[\ab]^\dims$  to $\RR^{\dims \times \ab}$, and the $i$-th row is the $t$-th standard basis vector given $s_i = t$. Then we define $w^* \in \RR^{\dims \times \ab}$ as follows:
\begin{align} 
&w^*(j,\cdot) = W_{1,j+1} (\posa, \cdot) -W_{1,j+1}(\posb, \cdot), \forall j \in [\dims-1] ;  \nonumber\\
&w^*(\dims,\cdot) = [\theta_1(\posa)-\theta_1(\posb),0,\cdots,0].\nonumber
\end{align}
Lemma~\ref{lem:PairwiseToLR} implies that $\forall t$, $\probof{Y^t=1} = \sigma\Paren{\langle w^*, X^t \rangle}$, where $\langle \cdot ,\cdot \rangle$ is the element-wise multiplication of matrices. According to the definition of the width of $\cD(\cW,\Theta)$, $\norm{w^*}\le \lambda \ab$. Now we can apply the sparse logistic regression method of Algorithm~\ref{alg:DPPair} to the samples in $S_{\posa,\posb}$.

Suppose $w^{priv}_{\posa,\posb}$ is the output of the private Frank-Wolfe algorithm, we define $U_{\posa, \posb} \in \RR^{\dims \times \ab}$ as follows: $\forall b \in [\ab]$,
\begin{align}
\label{equ:centering}
&U_{\posa, \posb} (j,b) = w^{priv}_{\posa, \posb}(j,b) - \frac1{\ab} \sum_{a \in [\ab]} w^{priv}_{\posa, \posb}(j,a), \forall j \in [\dims-1]; \nonumber\\
& U_{\posa, \posb} (\dims,b) = w^{priv}_{\posa, \posb}(\dims,b) + \frac1{\ab} \sum_{j \in [\dims-1]}\sum_{a \in [\ab]} w^{priv}_{\posa, \posb}(j,a).
\end{align}


$U_{\posa, \posb}$ can be seen as a ``centered" version of $w^{priv}_{\posa, \posb}$ (for the first $\dims-1$ rows).
It is not hard to see that $\langle U_{\posa, \posb}, x \rangle = \langle w^{priv}_{\posa, \posb}, x \rangle$, so $U_{\posa, \posb}$ is also a minimizer of the sparse logistic regression. 

For now, assume that $\forall  j \in [\dims-1], b \in [\ab]$, $U_{\posa, \posb}(j,b)$ is a  ``good'' approximation of $\Paren{ W_{1,j+1} (\posa, b) -W_{1,j+1}(\posb, b)}$, which we will show later. If we sum over $\posb \in [\ab]$, it can be shown that 
$\frac{1}{\ab}\sum_{\posb \in [\ab]}U_{\posa, \posb} ( j,b)$ is also a ``good'' approximation of $W_{1,j+1}(\posa, b)$, for all $ j \in [\dims-1]$, and $ \posa, b \in [\ab]$, because of the centering assumption of $\cW$, i.e., $\forall  j \in [\dims-1], b \in [\ab], \sum_{\posb \in [\ab]} W_{1,j+1}(\posb,b) = 0$. With these considerations in mind, we are able to introduce our algorithm.

\begin{algorithm}[H]

\KwIn{alphabet size $k$, $\ns$ i.i.d.\ samples $\{ z^1,\cdots, z^{\ns}\}$, where $z^m \in [\ab]^{\dims}$ for $m\in [\ns]$; an upper bound on $\lambda (\cW,\Theta)  \le \lambda$, privacy parameter $\rho$}

\For{$i=1$ to $\dims$}{
	\For{each pair $\posa \neq \posb \in [\ab]$}{
        $S_{\posa, \posb} \leftarrow \{ z^m, m\in [\ns]:  z_i^m \in \{u,v \} \}$
        
        $\forall z^m \in S_{\posa, \posb}$, $x^m \leftarrow \text{OneHotEncode}([z^{m}_{-i},1])$, $y^m \leftarrow 1$ if $ z_i^m=\posa$; $y^t \leftarrow -1$ if $ z_i^m=\posb$
   
        $w^{priv}_{\posa, \posb} \leftarrow \cA_{PFW}( D, \cL, \rho^{\prime}, \cC)$, where $\rho^{\prime} = \frac{\rho}{\ab^2\dims}$, $D=\{ \Paren{x^m, y^m}:z^m \in S_{\posa, \posb} \}$, $\cL(w;D) = \frac1{|S_{\posa, \posb}|} \sum_{m=1}^{|S_{\posa, \posb}|}  \log\Paren{1+e^{-y^m \langle w, x^m \rangle}}$, $\cC = \{ \norm{w}\le 2\lambda \ab\}$
   
      Define $U_{\posa, \posb} \in \RR^{\dims \times \ab}$ by centering the first $\dims-1$ rows of $w^{priv}_{\posa, \posb}$, as in Equation~\ref{equ:centering}}
	
	\For{$j \in [\dims] \backslash i$  and $\posa \in [\ab]$}{
$\widehat{W}_{i,j}(\posa,:) \leftarrow \frac{1}{\ab}\sum_{\posb\in [\ab]}U_{\posa, \posb} (\tilde{j},:)$, where $\tilde{j} = j$ when $j<i$ and $\tilde{j} = j-1$ when $j>i$
	      }
}
\KwOut{$\widehat{W}_{i,j} \in \RR^{\ab \times \ab}$ for all $i \neq j \in[\dims]$}

 \caption{Privately Learning Pairwise Graphical Model}
 \label{alg:DPPair}
\end{algorithm}

The following theorem is the main result of this section.
Its proof is structurally similar to that of Theorem~\ref{thm:est-ub-ising}.

\begin{theorem}
\label{thm:est-ub-pair}
  
Let $\cD(\cW,\Theta)$ be an unknown $\dims$-variable pairwise graphical model distribution, and we suppose that $\cD(\cW,\Theta)$ has width $\lambda(\cW,\Theta)\le \lambda$. There exists an efficient $\rho$-zCDP algorithm which outputs $\widehat{W}$ such that with probability greater than $2/3$,
$\absv{ W_{i,j}(\posa, \posb)- \widehat{W}_{i,j}(\posa, \posb)} \le \dist$, $\forall i \neq j \in [\dims], \forall \posa, \posb \in [\ab]$ if the number of i.i.d.\ samples satisfy
$$\ns = \Omega\Paren{ \frac{ \lambda^2 k^5 \log(\dims k) e^{O(\lambda)}}{\dist^4} + \frac{\sqrt{\dims} \lambda^2 k^{5.5} \log^2(\dims k)e^{O(\lambda)}}{\sqrt{\rho} \alpha^3}}.$$
\end{theorem}

\begin{proof}
We consider estimating $W_{1,j}$ for all $j \in [\dims]$ as an example. 
Fixing one pair $(\posa,\posb)$, let $S_{\posa,\posb}$ be the samples whose first element is either $\posa$ or $\posb$, and $\ns^{\posa,\posb}$ be the number of samples in $S_{\posa,\posb}$. 
We perform the following transformation on the samples in $S_{\posa,\posb}$: for the sample $Z$, let $Y=1$ if $Z_1 = \posa$, else $Y = -1$, and let $X$ be the one-hot encoding of the vector $[Z_{-1},1]$. 

Suppose the underlying joint distribution of $X$ and $Y$ is $P$, i.e., $(X,Y) \sim P$, then
by Theorem~\ref{thm:generalization_error}, when $\ns^{\posa,\posb} = O \Paren{\frac{\lambda^2 \ab^2 \log^2(\dims\ab)}{\gamma^2}+\frac{\sqrt{d} \lambda^2 \ab^3 \log^2(\dims\ab)}{\gamma^{\frac{3}{2}} \sqrt{\rho} } }$, with probability greater than $1-\frac1{10 \dims \ab^2}$, 
$$ \expectationf{\ell(U_{\posa, \posb};(X,Y))}{(X,Y)\sim P} - \expectationf{\ell(w^*;(X,Y))}{(X,Y)\sim P} \le \gamma.$$ 
The following lemma appears in~\cite{WuSD19}, which is analogous to Lemma~\ref{lem:parameter-error-ising} for the Ising model.
\begin{lemma}
\label{lem:parameter-error-kalphabet}
Let $\cD$ be a $\delta$-unbiased distribution on $[\ab]^{\dims-1}$. For $Z \sim \cD$, $X$ denotes the one-hot encoding of $Z$. 
Let $u_1, u_2 \in \RR^{(\dims-1) \times \ab}$ be two matrices where $\sum_{a} u_1 (i,a) =0$ and  $\sum_{a} u_2(i,a) =0$ for all $i \in [\dims-1]$.
Let $P $ be a distribution such that given $u_1, \theta_1\in \RR$, $\probof{Y=1| X= X} = \sigma\Paren{\langle u_1,x \rangle+\theta_1}$ for $(X,Y) \sim P$. Suppose $\expectationf{\log\Paren{1+e^{-Y \Paren{ \langle u_1, X \rangle+\theta_1}}} } {(X,Y)\sim P} - \expectationf{ \log\Paren{1+e^{-Y \Paren{ \langle u_2, X \rangle+\theta_2}}} }{(x,Y)\sim P} \le \gamma$ for $u_2\in \RR^{(\dims-1) \times \ab}, \theta_2\in \RR$, and $\gamma \le \delta e^{-2\norminfone{u_1} -2\norm{\theta_1}-6}$, then $\norminf{u_1-u_2} = O(e^{\norminfone{u_1} +\norm{\theta_1}} \cdot \sqrt{\gamma/\delta}).$
\end{lemma}

By Lemma~\ref{lem:marginal}, Lemma~\ref{lem:pairwise-unbiased} and Lemma~\ref{lem:parameter-error-kalphabet}, if we substitute $\gamma = \frac{e^{-6\lambda} \alpha^2}{\ab}$, when $\ns^{\posa,\posb}  = O\Paren{\frac{ \lambda^2 k^4 \log(\dims k) e^{O(\lambda)}}{\dist^4}+\frac{\sqrt{\dims} \lambda^2 k^{4.5} \log^2(\dims k)e^{O(\lambda)}}{\sqrt{\rho} \alpha^3}}, $
\begin{align}
\label{equation:approxdifference}
\absv{W_{1,j}(\posa, b) - W_{1,j}(\posb, b) - U^{\posa ,\posb} (j ,b)} \le \alpha, \forall j \in [\dims-1], \forall b \in [\ab].
\end{align}
  By a union bound, Equation~(\ref{equation:approxdifference}) holds for all $(\posa, \posb)$ pairs simultaneously with probability greater than $1-\frac1{10\dims}$.
If we sum over $\posb \in [\ab]$ and use the fact that $\forall j,b, \sum_{\posb \in [\ab]} W_{1,j}(\posb,b) = 0$, we have
$$\absv{W_{1,j}(\posa, b) -  \frac{1}{\ab}\sum_{\posb \in [\ab]}U_{\posa, \posb} ( j,b)} \le \alpha, \forall j \in [\dims-1], \forall \posa, b \in [\ab]. $$

Note that we need to guarantee that we obtain $\ns^{\posa, \posb}$ samples for each pair $(\posa, \posb)$. Since $\cD(\cW,\Theta)$ is $\delta$-unbiased, given $Z \sim \cD(\cW, \Theta)$, for all $\posa \neq \posb$,  $\probof{Z \in S_{\posa, \posb}} \ge 2\delta$. By Hoeffding's inequality, when $\ns = O\Paren{\frac{\ns^{\posa,\posb} }{\delta} + \frac{\log (\dims k^2)}{\delta^2}}$, with probability greater than $1 - \frac1{10\dims}$, we have enough samples for all $(\posa, \posb)$ pairs simultaneously. Substituting $\delta = \frac{e^{-6\lambda}}{k}$, we have
$$\ns = O\Paren{\frac{ \lambda^2 k^5 \log(\dims k) e^{O(\lambda)}}{\dist^4} +\frac{\sqrt{\dims} \lambda^2 k^{5.5} \log^2(\dims k)e^{O(\lambda)}}{\sqrt{\rho} \alpha^3}}.$$

The same argument holds for other entries of the matrix.
We conclude the proof by a union bound over $\dims$ iterations.

Finally, we note that the time compexity of the algorithm is $\poly(\ns, \dims)$ since the private Frank-Wolfe algorithm is time efficient by Corollary~\ref{cor:EmpiricalError}.
\end{proof}


\section{Privately Learning Binary $t$-wise MRFs}
\label{sec:bin-mrf}

Let $\cD$ be a $t$-wise MRF on $\{1, -1 \}^{\dims}$ with underlying dependency graph $G$ and factorization polynomial $h(x) = \sum_{I \in C_t(G)} h_I(x)$. 
We assume that the width of $\cD$ is bounded by $\lambda$, i.e., $\max_{i \in [\dims]} \norm{\partial_i h} \le \lambda$, where $\norm{h}\coloneqq \sum_{I \in C_t(G)} \absv{\bar{h}(I)}$. 
Similar to~\cite{KlivansM17}, given $\ns$ i.i.d.\ samples $ \{z^1, \cdots, z^{\ns}\}$ generated from an unknown distribution $\cD$, we consider the following two related learning objectives, under the constraint of $\rho$-zCDP:
\begin{enumerate}
\item[1.] find a multilinear polynomial $u$ such that with probability greater than $\frac{2}{3}$, $\norm{h-u} \coloneqq \sum_{I \in C_t(G)} \absv{\bar{h}(I) -\bar{u}(I)} \le \dist$ ;
\item[2.] find a multilinear polynomial $u$ such that with probability greater than $\frac{2}{3}$, for every maximal monomial $I$ of $h$, $\absv{\bar{h}(I) -\bar{u}(I)} \le \dist$.
\end{enumerate}


We note that our first objective can be viewed as parameter estimation in $\ell_1$ distance, where only an average performance guarantee is provided. 
In the second objective, the algorithm recovers every maximal monomial, which can be viewed as parameter estimation in $\ell_\infty$ distance. 
These two objectives are addressed in Sections~\ref{sec:MRFl_1} and~\ref{sec:MRFl_infty}, respectively.

\subsection{Parameter Estimation in $\ell_1$ Distance}
\label{sec:MRFl_1}
The following property of MRFs, from~\cite{KlivansM17}, plays a critical role in our algorithm. 
The proof is similar to that of Lemma~\ref{lem:IsingToLR}.
\begin{lemma}[Lemma 7.6 of~\cite{KlivansM17}]
\label{lem:sigmoidMRF}
Let $\cD$ be a $t$-wise MRF on $\{1, -1 \}^{\dims}$ with underlying dependency graph $G$ and factorization polynomial $h(x) = \sum_{I \in C_t(G)} h_I(x)$, then
$$\probof{Z_i=1|Z_{-i}=x} = \sigma(2\partial_i h(x)), \forall i \in [\dims], \forall x \in \{1 ,-1\}^{[\dims] \backslash i}.$$
\end{lemma}

 Lemma~\ref{lem:sigmoidMRF} shows that, similar to pairwise graphical models, it also suffices to learn the parameters of binary $t$-wise MRF using sparse logistic regression.

\begin{algorithm}[H]

\KwIn{ $\ns$ i.i.d.\ samples $\{ z^1,\cdots, z^{\ns}\}$, where $z^m \in \{\pm1\}^{\dims}$ for $m\in [\ns]$; an upper bound $\lambda$ on $\max_{i \in [\dims]} \norm{\partial_i h}$, privacy parameter $\rho$}

\For{$i=1$ to $\dims$}{
$\forall m \in [\ns]$, $x_m \leftarrow \Brack{\prod_{j \in I} z_j^m: I \subset [\dims \backslash i], \absv{I} \le t-1}$, $y_m \leftarrow z_i^m$

$w^{priv} \leftarrow \cA( D, \cL, \rho^{\prime}, \cC)$
where $D=\{\Paren{x_m, y_m} \}_{m=1}^{\ns}$, $\ell(w;d) =  \log\Paren{1+e^{-y \langle w, x \rangle}}$, $\cC = \{ \norm{w}\le 2\lambda\}$, and $\rho^{\prime} = \frac{\rho}{\dims}$

\For{$I \subset [\dims \backslash i]$ with $\absv{I} \le t-1$ }{
$\bar{u}(I \cup \{i\}) = \frac1{2} w^{priv}(I)$,  when $\arg\min (I \cup i) = i$
}
}

\KwOut{ $ \bar{u}(I): I \in C_t(K_{\dims})$,  where $K_{\dims}$ is the $\dims$-dimensional complete graph}

\caption{Private Learning binary $t$-wise MRF in $\ell_1$ distance}
\label{alg:DPMRF}
\end{algorithm}

\begin{theorem}
There exists a $\rho$-zCDP algorithm which, with probability at least $2/3$, finds a multilinear polynomial $u$ such that 
$\norm{h-u} \le \dist,$
given $\ns$ i.i.d.\ samples $Z^1,\cdots, Z^{\ns} \sim \cD$, where
$$\ns = O\Paren{ \frac{ (2t)^{O(t)} e^{O(\lambda t)} \cdot \dims^{4t} \cdot \log (\dims) }  {\dist^4} + \frac{ (2t)^{O(t)} e^{O(\lambda t)}\cdot \dims^{3t+\frac12} \cdot \log^2 (\dims)} { \sqrt{\rho}\dist^3}}. $$
\end{theorem}


\begin{proof}
Similar to the previous proof, we start by fixing $i=1$. Given a random sample $Z \sim \cD$,   let $X = \Brack{\prod_{j \in I} Z^j: I \subset [\dims] \backslash 1, \absv{I} \le t-1}$ and $Y = Z_i $. According to Lemma~\ref{lem:sigmoidMRF}, we know that $\expectation{Y|X} = \sigma\Paren{\langle w^*,X \rangle}$, where $w^* = \Paren{ 2\cdot \overline{ \partial_{1}h} (I): I \subset [\dims] \backslash 1, \absv{I}\le t-1}$. Furthermore, $\norm{w^*} \le 2\lambda$ by the width constraint. Now, given $\ns$ i.i.d.\ samples $\{ z^{m}\}_{m=1}^{\ns}$ drawn from $\cD$, it is easy to check that for any given $z^m$, its corresponding $(x^m,y^m)$ is one realization of $(X,Y)$. Let $w^{priv}$ be the output of  $\cA\Paren{ D, \cL, \frac{\rho}{\dims},  \{w: \norm{w}\le 2\lambda \}}$, where $D = \{(x^m,y^m)\}_{m=1}^{\ns}$ and $\ell(w;(x,y)) = \log\Paren{1+e^{-y \langle w, x \rangle}}$. By Lemma~\ref{thm:generalization_error}, $ \expectationf{\ell(w^{priv};(X,Y))}{Z\sim \cD(A,\theta)} - \expectationf{\ell(w^*;(X,Y))}{Z\sim \cD(A,\theta)} \le \gamma$ with probability greater than $1-\frac{1}{10\dims}$, assuming $\ns =\Omega\Paren{ \frac{ \sqrt{\dims} \lambda^2 \log^2(\dims) }{\sqrt{\rho} \gamma^{\frac{3}{2}}}+ \frac{ \lambda^2 \log(\dims)}{\gamma^2} }$.

Now we need the following lemma from~\cite{KlivansM17}, which is analogous to Lemma~\ref{lem:IsingToLR} for the Ising model.

\begin{lemma}[Lemma 6.4 of~\cite{KlivansM17}]
\label{lem:parameter-error-mrf}
Let $P$ be a distribution on $\{-1,1\}^{\dims-1} \times \{-1,1\}$. Given multilinear polynomial $u_1\in \RR^{\dims-1}$, $\probof{Y=1|X=x} = \sigma\Paren{u_1(X)}$ for $(X,Y) \sim P$. 
Suppose the marginal distribution of $P$ on $X$ is $\delta$-unbiased, 
and $\expectationf{\log\Paren{1+e^{-Y \Paren{ u_1(X) }}} } {(X,Y)\sim P} - \expectationf{ \log\Paren{1+e^{-Y \Paren{ u_2(X)}}} }{(X,Y)\sim P} \le \gamma$ for another multilinear polynomial $u_2$, where $\gamma \le \delta^t e^{-2\norm{u_1}-6}$, then $\norm{u_1-u_2} = O\Paren{ (2t)^{t} e^{\norm{u_1}} \cdot \sqrt{\gamma/\delta^t}\cdot {\dims \choose t} }.$
\end{lemma}

By substituting $\gamma = e^{ -O( \lambda t)} \cdot (2t)^{-O(t)} \cdot \dims^{-3t} \cdot \alpha^2 $, we have that with probability greater than $1-\frac1{10\dims}$, $\sum_{I:\arg\min I =1} \absv{\bar{u}(I) - h(I)} \le \frac{\alpha}{\dims}$. We note that the coefficients of different monomials are recovered in each iteration. Therefore, by a union bound over $\dims$ iterations, we prove the desired result.
\end{proof}

\subsection{Parameter Estimation in $\ell_{\infty}$ Distance}
\label{sec:MRFl_infty}
In this section, we introduce a slightly modified version of the algorithm in the last section.

\begin{algorithm}[H]

\KwIn{ $\ns = \ns_1+\ns_2$ i.i.d.\ samples $\{ z^1,\cdots, z^{\ns}\}$, where $z^m \in \{\pm1\}^{\dims}$ for $m\in [\ns]$; an upper bound $\lambda$ on $\max_{i \in [\dims]} \norm{\partial_i h}$, privacy parameter $\rho$}

\For{each $I \subset [\dims]$ with $\absv{I} \le t$}{
Let $Q(I) \coloneqq \frac1{n_2} \sum_{m=n_1}^{n_1+n_2} \prod_{j \in I} z_j^m$

Compute $\hat{Q}(I)$, an estimate of  $Q(I)$ through an $\rho/2$-zCDP query release algorithm (PMW~\cite{HardtR10} or sepFEM~\cite{neworacle})
}

\For{$i=1$ to $\dims$}{
$\forall m \in [\ns_1]$, $x_m \leftarrow \Brack{\prod_{j \in I} z^m_j: I \subset [\dims \backslash i], \absv{I} \le t-1}$, $y_m \leftarrow z_i^m$

$w^{priv} \leftarrow \cA( D, \cL, \rho^{\prime}, \cC)$
where $D=\{\Paren{x_m, y_m} \}_{m=1}^{\ns}$, $\ell(w;d) =  \log\Paren{1+e^{-y \langle w, x \rangle}}$, $\cC = \{ \norm{w}\le 2\lambda\}$, and $\rho^{\prime} = \frac{\rho}{2\dims}$

Define a polynomial $v_i: \RR^{\dims-1} \rightarrow \RR$ by setting $\bar{v_i}(I) = \frac1{2} w^{priv}(I)$ for all $I \subset [\dims \backslash i]$

\For{each $I \subset [\dims \backslash i]$ with $\absv{I} \le t-1$}{
$\bar{u}(I \cup \{i\}) = \sum_{I^{\prime} \subset [\dims]}  \overline{\partial_I v_i}(I^{\prime}) \cdot \hat{Q}(I^{\prime})$,
when $\arg\min (I \cup i) = i$
}
}

\KwOut{ $ \bar{u}(I): I \in C_t(K_{\dims})$,  where $K_{\dims}$ is the $\dims$-dimensional complete graph}

\caption{Private Learning binary $t$-wise MRF in $\ell_\infty$ distance}
\label{alg:DPMRF-infty}
\end{algorithm}

We first show that if the estimates $\hat Q$ for the parity queries
$Q$ are sufficiently accurate, Algorithm~\ref{alg:DPMRF-infty} solves
the $\ell_\infty$ estimation problem, as long as the sample size $n_1$
is large enough.

\begin{lemma}\label{yo}
  Suppose that the estimates $\hat Q$ satisfies
  $|\hat Q(I) - Q(I)| \leq \alpha/(8\lambda)$ for all $I\subset [p]$
  such that $|I|\leq t$ and $n_2 = \Omega(\lambda^2 t\log(p)/\alpha^2)$. Then
  with probability at least $3/4$, Algorithm~\ref{alg:DPMRF-infty}
  outputs a multilinear polynomial $u$ such that for every maximal
  monomial $I$ of $h$, $\absv{\bar{h}(I) - \bar{u}(I)} \le \dist,$
  given $\ns$ i.i.d.\ samples $Z^1,\cdots, Z^{\ns} \sim \cD$, as long
  as
$$n_1 =\Omega\Paren{ \frac{ e^{5\lambda t} \cdot \sqrt{\dims} \log^2(\dims)
  }{\sqrt{\rho} \dist^{\frac{9}{2}}}+ \frac{ e^{6\lambda t} \cdot
    \log(\dims)}{\dist^6} }. $$
\end{lemma}

\begin{proof}
  We will condition on the event that $\hat{Q}$ is a ``good'' estimate
  of $Q$: $|\hat Q(I) - Q(I)| \leq \alpha/(8\lambda)$ for all
  $I\subset [p]$ such that $|I|\leq
  t$. 
  Let us fix $i=1$. Let
  $X = \Brack{\prod_{j \in I} Z^j: I \subset [\dims] \backslash \{1
    \}, \absv{I} \le t-1}$, $Y = Z_i $, and we know that
  $\expectation{Y|X} = \sigma\Paren{\langle w^*,X \rangle}$, where
  $w^* = \Paren{ 2\cdot \overline{ \partial_{1}h} (I): I \subset
    [\dims] \backslash 1, \absv{I}\le t-1}$. Now given $\ns_1$
  i.i.d.\ samples $\{ z^{m}\}_{m=1}^{\ns_1}$ drawn from $\cD$, let
  $w^{priv}$ be the output of
  $\cA\Paren{ D, \cL, \frac{\rho}{\dims}, \{w: \norm{w}\le 2\lambda
    \}}$, where $D = \{(x^m,y^m)\}_{m=1}^{\ns_1}$ and
  $\ell(w;(x,y)) = \log\Paren{1+e^{-y \langle w, x
      \rangle}}$. Similarly, with probability at least
  $1-\frac{1}{10\dims}$,
  $$ \expectationf{\ell(w^{priv};(X,Y))}{Z\sim \cD(A,\theta)} -
  \expectationf{\ell(w^*;(X,Y))}{Z\sim \cD(A,\theta)} \le \gamma$$ as long as
  $\ns_1 =\Omega\Paren{ \frac{ \sqrt{\dims} \lambda^2 \log^2(\dims)
    }{\sqrt{\rho} \gamma^{\frac{3}{2}}}+ \frac{ \lambda^2
      \log(\dims)}{\gamma^2} }$.

Now we utilize Lemma 6.4 from~\cite{KlivansM17}, which states that if  $ \expectationf{\ell(w^{priv};(X,Y))}{Z\sim \cD(A,\theta)} - \expectationf{\ell(w^*;(X,Y))}{Z\sim \cD(A,\theta)} \le \gamma$, given a random sample $X$, for any maximal monomial $I \subset [\dims]\backslash \{1\}$ of $\partial_1 h$,
$$\probof{  \absv{\overline{\partial_1 h}(I) -  \partial_I v_1(X)} \ge \frac{\dist}{4} } < O\Paren{ \frac{\gamma \cdot e^{3\lambda t}}{\dist^2}}. $$

By replacing
$\gamma = \frac{ e^{-3\lambda t}\cdot \dist^3}{8\lambda}$, we have
$\probof{ \absv{\overline{\partial_1 h}(I) - \partial_I v_1(X)} \ge
  \frac{\dist}{4} } <\frac{\dist}{8\lambda}$, as long as
$n_1 = \Omega\Paren{ \frac{ \sqrt{\dims} e^{5\lambda t} \log^2(\dims)
  }{\sqrt{\rho} \dist^{\frac{9}{2}}}+ \frac{ e^{6\lambda t}
    \log(\dims)}{\dist^6} }$.  Accordingly, for any maximal monomial
$I$,
$ \absv{\expectation{\partial_I v_1(X)} - \overline{\partial_1 h} (I)}
\le \expectation { \absv{\partial_I v_1(X) - \overline{\partial_1
      h}(I) }} \le \frac{\dist}{4} + 2\lambda \cdot
\frac{\dist}{8\lambda} = \frac{\dist}{2}$.  By Hoeffding inequality,
given $n_2 = \Omega\Paren{\frac{\lambda^2t\log {\dims}}{\dist^2}}$, for each maximal
monomial $I$, with probability greater than $1-\frac{1}{\dims^t}$,
$\absv{ \frac1{n_2} \sum_{m=1}^{n_2}\partial_I v_1(X_m) -
  \expectation{\partial_I v_1(X)} }\le \frac{\dist}{4}$. Note that
$\absv{ Q(I) - \hat{Q}(I)} \le \frac{\dist}{8\lambda}$, then
$\absv{\frac1{n_2} \sum_{m=1}^{n_2}\partial_I v_1(X_m)
  -\sum_{I^{\prime} \subset [\dims]} \overline{\partial_I
    v_1}(I^{\prime}) \cdot \hat{Q}(I^{\prime}) } \le \frac{\dist}{8}
$. Therefore,
\begin{align}
&\absv{\sum_{I^{\prime} \subset [\dims]}  \overline{\partial_I v_1}(I^{\prime}) \cdot \hat{Q}(I^{\prime})   -   \overline{\partial_1 h}(I) } \nonumber\\
\le&\absv{ \sum_{I^{\prime} \subset [\dims]}  \overline{\partial_I v_1}(I^{\prime}) \cdot \hat{Q}(I^{\prime})   -\frac1{n_2} \sum_{m=1}^{n_2}\partial_I v_1(X_m) } + \absv{\frac1{n_2} \sum_{m=1}^{n_2}\partial_I v_1(X_m) -   \expectation{\partial_I v_1(X)}  } + \absv{  \expectation{\partial_I v_1(X)} -  \overline{\partial_1 h}(I)   } \nonumber\\
\le&\frac{\dist}{8}+\frac{\dist}{4}+\frac{\dist}{2} = \frac{7\dist}{8}.\nonumber
\end{align}
Finally, by a union bound over $\dims$ iterations and all the maximal monomials, we prove the desired results.\end{proof}

We now consider two private algorithms for releasing the parity
queries. The first algorithm is called Private Multiplicative Weights
(PMW)~\cite{HardtR10}, which provides a better accuracy guarantee but
runs in time exponential in the dimension $\dims$. The following theorem can be viewed as a zCDP version of Theorem 4.3 in~\cite{Vadhan17}, 
by noting that during the analysis, every iteration satisfies $\eps_0$-DP, which naturally satisfies $\eps_0^2$-zCDP, and by replacing the strong composition theorem of $(\eps,\delta)$-DP by the composition theorem of zCDP (Lemma~\ref{lem:dpcomp}).

\begin{lemma}[Sample complexity of PMW, modification of Theorem 4.3 of~\cite{Vadhan17}]\label{hey}
  The PMW algorithm satisfies $\rho$-zCDP and releases $\hat{Q}$ such
  that with probability greater than $\frac{19}{20}$, for all
  $I \subset [\dims]$ with $\absv{I} \le t$,
  $\absv{\hat{Q}(I) -Q(I)} \le \frac{\dist}{8\lambda}$ as long as the
  size of the data set
  $$n_2 = \Omega\Paren{ \frac{ t \lambda^2 \cdot \sqrt{\dims}
      \log{\dims}}{\sqrt{\rho}\dist^2} }.$$
\end{lemma}

The second algorithm sepFEM (Separator-Follow-the-perturbed-leader
with exponential mechanism) has slightly worse sample complexity, but
runs in polynomial time when it has access to an optimization oracle
$\mathcal{O}$ that does the following: given as input a weighted
dataset
$(I_1, w_1), \ldots , (I_m, w_m)\in 2^{[p]} \times \mathbb{R}$, find $x\in \{\pm 1\}^\dims$,
\[
  \max_{x\in \{\pm 1\}^\dims} \sum_{i=1}^m w_i \prod_{j\in I_i} x_{j}.
\]
The oracle $\mathcal{O}$ essentially solves cost-sensitive
classification problems over the set of parity functions~\cite{ZLA03},
and it can be implemented with an integer program
solver~\cite{neworacle,GaboardiAHRW14}.

\begin{lemma}[Sample complexity of sepFEM,~\cite{neworacle}]\label{man}
  The sepFEM algorithm satisfies $\rho$-zCDP and releases $\hat{Q}$
  such that with probability greater than $\frac{19}{20}$, for all
  $I \subset [\dims]$ with $\absv{I} \le t$,
  $\absv{\hat{Q}(I) -Q(I)} \le \frac{\dist}{8\lambda}$ as long as the
  size of the data set
  $$n_2 = \Omega\Paren{ \frac{ t \lambda^2 \cdot {\dims^{5/4}}
      \log{\dims}}{\sqrt{\rho}\dist^2} }$$ The algorithm runs in
  polynomial time given access to the optimization oracle
  $\mathcal{O}$ defined above.
\end{lemma}

Now we can combine Lemmas~\ref{yo}, \ref{hey}, and \ref{man} to state
the formal guarantee of Algorithm~\ref{alg:DPMRF-infty}.

\begin{theorem}
  Algorithm~\ref{alg:DPMRF-infty} is a $\rho$-zCDP algorithm which, with probability at least $2/3$, 
  finds a multilinear polynomial $u$ such that for every maximal
  monomial $I$ of $h$, $\absv{\bar{h}(I) - \bar{u}(I)} \le \dist,$
  given $\ns$ i.i.d.\ samples $Z^1,\cdots, Z^{\ns} \sim \cD$, and
  \begin{enumerate}
  \item if it uses PMW for releasing $\hat Q$; it has a sample complexity of
$$\ns =  O\Paren{ \frac{ e^{5\lambda t} \cdot \sqrt{\dims} \log^2(\dims)
  }{\sqrt{\rho} \dist^{\frac{9}{2}}}+ \frac{ t \lambda^2 \cdot
    \sqrt{\dims} \log{\dims}}{\sqrt{\rho}\dist^2}+\frac{ e^{6\lambda t} \cdot
    \log(\dims)}{\dist^6} }$$ 
and a runtime complexity that is exponential in $\dims$;
\item if it uses sepFEM for releasing $\hat Q$, it has a sample
  complexity of
$$\ns =\tilde O\Paren{ \frac{ e^{5\lambda t} \cdot \sqrt{\dims} \log^2(\dims)
  }{\sqrt{\rho} \dist^{\frac{9}{2}}}+ \frac{ t \lambda^2 \cdot
    \dims^{5/4} \log{\dims}}{\sqrt{\rho} \dist^2}+\frac{ e^{6\lambda
      t} \cdot \log(\dims)}{\dist^6} }$$ and runs in polynomial time
whenever $t = O(1)$.
\end{enumerate}
\end{theorem}


\section{Lower Bounds for Parameter Learning}
\label{sec:est-lb}
The lower bound for parameter estimation is based on mean estimation in $\ell_\infty$ distance.

\begin{theorem}
\label{thm:est-lb}
  Suppose $\cA$ is an $(\eps, \delta)$-differentially private
  algorithm that takes $n$ i.i.d.\ samples $Z^1, \ldots , Z^n$ drawn
  from any unknown $\dims$-variable Ising model $\cD(A,\theta)$ and outputs $\hat A$ such that
$\expectation {  \max_{i,j \in [p]} |A_{i, j} - \hat A_{i, j}|} \leq \alpha \leq 1/50.$
  Then $\ns = \Omega\Paren{\frac{\sqrt{\dims}}{\alpha\eps}}$.
\end{theorem}

\begin{proof}
  Consider a Ising model $\cD (A, 0)$ with
  $A \in \RR^{\dims \times \dims}$ defined as follows: for
  $i \in [\frac{\dims}{2}], A_{2i-1,2i} = A_{2i,2i-1} = \eta_i\in
  [-\ln(2), \ln(2)]$, and $A_{ll'} = 0$ for all other pairs of
  $(l, l')$. This construction divides the $\dims$ nodes
  into $\frac{\dims}{2}$ pairs, where there is no correlation between nodes belonging to different pairs.
 It follows that
\begin{align}
  \probof {Z_{2i-1}=1, Z_{2i} =1 } =\probof {Z_{2i-1}=-1,  Z_{2i} = -1} = \frac{1}{2} \frac{e^{\eta_i}}{e^{\eta_i}+1},\nonumber\\
  \probof {Z_{2i-1}=1, Z_{2i} =-1 } =\probof {Z_{2i-1}=-1,  Z_{2i} = 1} =  \frac{1}{2} \frac{1}{e^{\eta_i}+1}.\nonumber
\end{align}
For each observation $Z$, we obtain an observation
$X\in \{\pm 1\}^{\dims/2}$ such that $X_{i} = Z_{2i-1} Z_{2i}$.  Then
each observation $X$ is distributed according to a product
distribution in $\{\pm 1\}^{(\dims/2)}$ such that the mean of each
coordinate $j$ is $(e^{\eta_i} - 1)/(e^{\eta_i} + 1)\in [-1/3, 1/3]$.

Suppose that an $(\eps, \delta)$-differentially private algorithm
takes $n$ observations drawn from any such Ising model distribution
and output a matrix $\hat A$ such that
$\E \left[\max_{i,j\in [p]} |A_{i,j} - \hat A_{i, j}|\right] \leq \alpha$. Let
$\hat \eta_i = \min\{\max\{\hat A_{2i-1, 2i} , -\ln(2)\}, \ln(2)\}$ be
the value of $A_{2i-1, 2i}$ rounded into the range of
$[-\ln(2), \ln(2)]$, and so $|\eta_i - \hat \eta_i| \leq \alpha$. It
follows that
\begin{align*}
  \left| \frac{e^{\eta_i} - 1}{e^{\eta_i} + 1} -  \frac{e^{\hat \eta_i} - 1}{e^{\hat \eta_i} + 1}\right| 
                                                                                                         &= 2  \left| \frac{e^{\eta_i} - e^{\hat \eta_i}}{(e^{\eta_i} + 1)(e^{\hat \eta_i} + 1)}\right|\\
                                                                                                         &< 2 \left| {e^{\eta_i} - e^{\hat \eta_i}}\right| \leq 4 \left(e^{|\eta_i - \hat \eta_i|} - 1\right)\leq 8 |\eta_i - \hat \eta_i| 
\end{align*}
where the last step follows from the fact that $e^a \leq 1 + 2a$ for
any $a\in [0, 1]$. Thus, such private algorithm also can estimate the
mean of the product distribution accurately:
\[
  \E \left[\sum_{i=1}^{p/2} \left| \frac{e^{\eta_i} - 1}{e^{\eta_i} +
        1} - \frac{e^{\hat \eta_i} - 1}{e^{\hat \eta_i} + 1}\right|^2
  \right] \leq 32 p \alpha^2
\]
Now we will use the following sample complexity lower bound on private
mean estimation on product distributions.
\begin{lemma}[Lemma 6.2 of~\cite{KamathLSU19}]\label{prod-low}
  If $M\colon \{\pm 1\}^{n\times d} \rightarrow [-1/3, 1/3]^d$ is
  $(\eps, 3/(64n))$-differentially private, and for every product
  distribution $P$ over $\{\pm 1\}^d$ such that the mean of each
  coordinate $\mu_j$ satisfies $-1/3 \leq \mu_j \leq 1/3$,
  \[
    \E_{X \sim P^n} \left[ \| M(X) - \mu \|_2^2 \right] \leq \gamma^2
    \leq \frac{d}{54},
  \]
  then $n\geq d/(72 \gamma \eps)$.
\end{lemma}
Then our stated bound follows by instantiating $\gamma^2 = 32p\alpha^2$
and $d = p/2$ in Lemma~\ref{prod-low}.\end{proof}


\section{Structure Learning of Graphical Models}
\label{sec:struct-ub}
In this section, we will give an $(\varepsilon,\delta)$-differentially private algorithm for learning the \emph{structure} of a Markov Random Field.
The dependence on the dimension $d$ will be only \emph{logarithmic}, in comparison to the complexity of privately learning the parameters.
As we have shown in Section~\ref{sec:est-lb}, this dependence is necessarily polynomial in $\dims$, even under approximate differential privacy.
Furthermore, as we will show in Section~\ref{sec:struct-lb}, if we wish to learn the structure of an MRF under more restrictive notions of privacy (such as pure or concentrated), the complexity also becomes polynomial in $\dims$.
Thus, in very high-dimensional settings, learning the structure of the MRF under approximate differential privacy is essentially the only notion of private learnability which is tractable.

The following lemma is immediate from stability-based mode arguments (see, e.g., Proposition 3.4 of~\cite{Vadhan17}).
\begin{lemma}
  Suppose there exists a (non-private) algorithm which takes $X = (X^1, \dots, X^n)$ sampled i.i.d.\ from some distribution $\cD$, and outputs some fixed value $Y$ (which may depend on $\cD$) with probability at least $2/3$.
  Then there exists an $(\varepsilon, \delta)$-differentially private algorithm which takes $O\left(\frac{n\log(1/\delta)}{\varepsilon}\right)$ samples and outputs $Y$ with probability at least $1 - \delta$.
\end{lemma}

We can now directly import the following theorem from~\cite{WuSD19}.
\begin{theorem}[\cite{WuSD19}]
There exists an algorithm which, with probability at least $2/3$, learns the structure of a pairwise graphical model.
  It requires $n = O\left(\frac{\lambda^2 \ab^4 e^{14\lambda} \log(\dims \ab)}{\eta^4}\right)$ samples.
\end{theorem}

This gives us the following private learning result as a corollary.
\begin{corollary}
\label{thm:str-ub-pair}
  There exists an $(\varepsilon, \delta)$-differentially private algorithm which, with probability at least $2/3$, learns the structure of a pairwise graphical model.
  It requires $n = O\left(\frac{\lambda^2 \ab^4 e^{14\lambda} \log(\dims \ab)\log(1/\delta)}{\varepsilon\eta^4}\right)$ samples.
\end{corollary}

For binary MRFs of higher-order, we instead import the following theorem from~\cite{KlivansM17}:

\begin{theorem}[\cite{KlivansM17}]
There exists an algorithm which, with probability at least $2/3$, learns the structure of a binary $t$-wise MRF.
It requires $n = O\left(\frac{ e^{O\Paren{ \lambda t}} \log(\frac{\dims}{\eta} )} {\eta^4}\right)$ samples.
\end{theorem}

This gives us the following private learning result as a corollary.
\begin{corollary}
  There exists an $(\varepsilon, \delta)$-differentially private algorithm which, with probability at least $2/3$,  learns the structure of a binary $t$-wise MRF.
  It requires $$n = O\left(\frac{ e^{O\Paren{ \lambda t}} \log(\frac{\dims}{\eta} )\log(1/\delta)} {\varepsilon \eta^4}\right)$$ samples.
\end{corollary}

\section{Lower Bounds for Structure Learning of Graphical Models}
\label{sec:struct-lb}

%
%
%

In this section, we will prove structure learning lower bounds under pure DP or zero-concentrated DP. 
The graphical models we consider are the Ising models and pairwise graphical model. 
However, we note that all the lower bounds for the Ising model also hold for binary $t$-wise MRFs, since the Ising model is a special case of binary $t$-wise MRFs corresponding to $t=2$. 
We will show that under $(\eps,0)$-DP or $\rho$-zCDP, a polynomial dependence on the dimension is unavoidable in the sample complexity.

In Section~\ref{sec:struct-lb-Ising}, we assume that our samples are generated from an Ising model. In Section~\ref{sec:struct-lb-pairwise}, we extend our lower bounds to pairwise graphical models.

\subsection{Lower Bounds for Structure Learning of Ising Models}
\label{sec:struct-lb-Ising}

\begin{theorem}
\label{thm:str-ising}
Any $(\eps,0)$-DP algorithm which learns the structure of an Ising model with minimum edge weight $\eta$ with probability at least $2/3$ requires
$\ns = \Omega\Paren{\frac{\sqrt{\dims}}{\eta\eps} + \frac{\dims}{\eps}}$  samples. 
Furthermore, at least $\ns =\Omega\Paren{ \sqrt {\frac{\dims} {\rho}} }$ samples are required for the same task under $\rho$-zCDP.
\end{theorem}

\begin{proof}
Our lower bound argument is in two steps. 
The first step is to construct a set of distributions, consisting of $2^{\frac{\dims}{2}}$ different Ising models such that any feasible structure learning algorithm should output different answers for different distributions. In the second step, we utilize the probabilistic packing argument for DP~\cite{AcharyaSZ20}, or the packing argument for zCDP~\cite{BunS16} to get the desired lower bound.

To start, we would like to use the following binary code to construct the distribution set. Let $\cC = \{ 0,1\}^{\frac{\dims}{2}}$,
given $c \in C $, we construct the corresponding distribution $\cD (A^c, 0)$ with $A^c \in \RR^{\dims \times \dims}$ defined as follows: for $i \in  [\frac{\dims}{2}], A^c_{2i-1,2i} = A^c_{2i,2i-1} = \eta \cdot c[i]$, and 0 elsewhere. By construction, we divide the $\dims$ nodes into $\frac{\dims}{2}$ different pairs, where there is no correlation between nodes belonging to different pairs. Furthermore, for pair $i$, if $c[i] = 0$, which means the value of node $2i-1$ is independent of node $2i$, it is not hard to show 
\begin{align}
\probof {Z_{2i-1}=1, Z_{2i} =1 } =\probof {Z_{2i-1}=-1,  Z_{2i} = -1} = \frac1{4},\nonumber\\
\probof {Z_{2i-1}=1, Z_{2i} =-1 } =\probof {Z_{2i-1}=-1,  Z_{2i} = 1} = \frac1{4}.\nonumber
\end{align}
On the other hand, if $c[i]=1$,
\begin{align}
\probof {Z_{2i-1}=1, Z_{2i} =1 } =\probof {Z_{2i-1}=-1,  Z_{2i} = -1} = \frac1{2} \cdot \frac{e^{\eta}}{e^{\eta}+1},\nonumber\\
\probof {Z_{2i-1}=1, Z_{2i} =-1 } =\probof {Z_{2i-1}=-1,  Z_{2i} = 1} = \frac1{2} \cdot \frac{1}{e^{\eta}+1}.\nonumber
\end{align}
The Chi-squared distance between these two distributions is 
$$8 \Brack{\Paren{\frac12\cdot\frac{e^\eta}{e^\eta+1} -\frac14}^2+ \Paren{\frac12\cdot\frac{1}{e^\eta+1} -\frac14}^2} = \Paren{1-\frac{2}{e^\eta+1}}^2 \le 4\eta^2.$$

Now we want to upper bound the total variation distance between $\cD (A^{c_1}, 0)$ and $\cD (A^{c_2}, 0)$ for any $c_1 \neq c_2 \in \cC$. Let $P_i$ and $Q_i$ denote the joint distribution of node $2i-1$ and node $2i$ corresponding to $\cD (A^{c_1}, 0)$ and $\cD (A^{c_2}, 0)$.
We have that
$$d_{TV} \Paren{\cD (A^{c_1}, 0),\cD (A^{c_2}, 0)} \le \sqrt{ 2 d_{KL} \Paren{\cD (A^{c_1}, 0),\cD (A^{c_2}, 0)} }=  \sqrt{ 2 \sum_{i=1}^{\frac{\dims}{2}}d_{KL} \Paren{P_i, Q_i}} \le \min\Paren{2\eta\sqrt{\dims}, 1},$$
where the first inequality is by Pinsker's inequality, and the last inequality comes from the fact that the KL divergence is always upper bounded by the Chi-squared distance.

In order to attain pure DP lower bounds, we utilize the probabilistic version of the packing argument in~\cite{hz}, as stated below.
\begin{lemma}
\label{lem:coupling}
Let $\cV=\{P_1, P_2,...,P_M\}$ be a set of $M$ distributions over $\cX^{\ns}$. Suppose that for any pair of distributions $P_i$ and $P_j$, there exists a coupling between $P_i$ and $P_j$, such that $\expectation{\ham {\Xon} {\Yon}} \le D$, where $\Xon\sim P_i$ and $\Yon\sim P_j$. Let $ \{S_{i}\}_{i \in [M]}$ be a collection of disjoint subsets of $\cS$. If there exists an $\eps$-DP algorithm $\cA : \cX^{\ns} \to \cS$ such that for every $i \in [M]$, given $\Zon \sim P_i$, $\probof{\cA(\Zon) \in S_{i}}\ge \frac{9}{10}$, then
$$\eps = \Omega \Paren{\frac{\log M}{D}}.$$
\end{lemma}

For any $c_1, c_2 \in C$, we have $d_{TV} \Paren{\cD (A^{c_1}, 0),\cD (A^{c_2}, 0)} \le \min\Paren{2\eta\sqrt{\dims}, 1}$. By the property of maximal coupling~\cite{coupling}, there must exist some coupling between $\cD^{\ns} (A^{c_1}, 0)$ and $\cD^{\ns} (A^{c_2}, 0)$  with expected Hamming distance smaller than $\min\Paren{2\ns\eta\sqrt{\dims} , \ns}$. Therefore, we have $\eps = \Omega \Paren{\frac{\log \absv{\cC}}{\min\Paren{\ns\eta\sqrt{\dims} , \ns}}}$, and  accordingly, $\ns = \Omega\Paren{\frac{\sqrt{\dims}}{\eta\eps} + \frac{\dims}{\eps}}$.

Now we move to zCDP lower bounds. We utilize a different version of the  packing argument~\cite{BunS16}, which works under zCDP.
\begin{lemma}
\label{lem:coupling-zCDP}
Let $\cV=\{P_1, P_2,...,P_M\}$ be a set of $M$ distributions over $\cX^{\ns}$. Let $ \{S_{i}\}_{i \in [M]}$ be a collection of disjoint subsets of $\cS$. If there exists an $\rho$-zCDP algorithm $\cA : \cX^{\ns} \to \cS$ such that for every $i \in [M]$, given $\Zon \sim P_i$, $\probof{\cA(\Zon) \in S_{i}}\ge \frac{9}{10}$, then
$$\rho = \Omega \Paren{\frac{\log M}{\ns^2}}.$$
\end{lemma}
By Lemma~\ref{lem:coupling-zCDP}, we derive $\rho = \Omega\Paren{\frac{\dims}{\ns^2}}$ and $\ns =\Omega\Paren{ \sqrt {\frac{\dims} {\rho}} }$ accordingly.
\end{proof}

\subsection{Lower Bounds for Structure Learning of Pairwise Graphical Models}
\label{sec:struct-lb-pairwise}

Similar techniques can be used to derive lower bounds for pairwise graphical models.

\begin{theorem}
Any $(\eps,0)$-DP algorithm which learns the structure of the $\dims$-variable pairwise graphical models with minimum edge weight $\eta$ with probability at least $2/3$ requires
$\ns = \Omega\Paren{\frac{\sqrt{\dims}}{\eta\eps} + \frac{\ab^2\dims}{\eps}}$ samples. Furthermore, at least $\ns =\Omega\Paren{  \sqrt {\frac{\ab^2\dims} {\rho}} }$ samples are required for the same task under $\rho$-zCDP.
\end{theorem}

\begin{proof}
Similar to before, we start with constructing a distribution set consisting of $2^{O\Paren{\ab\dims}}$ different pairwise graphical models such that any accurate structure learning algorithm must output different answers for different distributions.

Let $\cC$ be the real symmetric matrix with each value constrained to either $0$ or $\eta$, i.e., $\cC = \{W \in \{0, \eta \}^{\ab \times \ab}: W = W^T\}$. 
Without loss of generality, we assume $\dims$ is even. 
Given $c = [ c_1, c_2,\cdots, c_{\dims}]$, where $ c_1, c_2,\cdots, c_{\dims} \in C $, we construct the corresponding distribution $\cD (\cW^c, 0)$ with $\cW^c$ defined as follows: for $l \in  [\frac{\dims}{2}], W^c_{2l-1,2l} = c_l$, and for other pairs $(i,j)$, $W^c_{i,j} =  0$. Similarly, by this construction we divide $\dims$ nodes into $\frac{\dims}{2}$ different pairs, and there is no correlation between nodes belonging to different pairs.

We first prove lower bounds under $(\eps,0)$-DP.  By Lemma~\ref{lem:coupling}, $\eps = \Omega\Paren{\frac{\log \absv{\cC}}{\ns}}$, since for any two $\ns$-sample distributions, the expected coupling distance can be always upper bounded by $\ns$. We also note that $\absv{\cC} = \Paren{ 2^{\frac{\ab (\ab+1)}{2}}}^\dims$. Therefore, we have $\ns = \Omega\Paren{\frac{\ab^2\dims}{\eps}}$. At the same time, $\ns = \Omega\Paren{\frac{\sqrt{\dims}}{\eta\eps}}$ is another lower bound, inherited from the easier task of learning Ising models.

With respect to zCDP, we utilize Lemma~\ref{lem:coupling-zCDP} and obtain $\rho = \Omega\Paren{\frac{\ab^2 \dims}{\ns^2}}$. 
Therefore, we have $\ns =\Omega\Paren{  \sqrt {\frac{\ab^2\dims} {\rho}} }$.
\end{proof}
%
%
%


\section*{Acknowledgments}
The authors would like to thank Kunal Talwar for suggesting the study of this problem, and Adam Klivans, Frederic Koehler, Ankur Moitra, and Shanshan Wu for helpful and inspiring conversations.
\begin{CJK*}{UTF8}{gbsn}
GK would like to thank Chengdu Style Restaurant (古月飘香) in Berkeley for inspiration in the conception of this project.
\end{CJK*}

\bibliographystyle{alpha}
\bibliography{biblio}

\newcommand{\etalchar}[1]{$^{#1}$}
\begin{thebibliography}{MdCCU16}

\bibitem[AKN06]{AbbeelKN06}
Pieter Abbeel, Daphne Koller, and Andrew~Y. Ng.
\newblock Learning factor graphs in polynomial time and sample complexity.
\newblock {\em Journal of Machine Learning Research}, 7(Aug):1743--1788, 2006.

\bibitem[AKSZ18]{AcharyaKSZ18}
Jayadev Acharya, Gautam Kamath, Ziteng Sun, and Huanyu Zhang.
\newblock Inspectre: Privately estimating the unseen.
\newblock In {\em Proceedings of the 35th International Conference on Machine
  Learning}, ICML '18, pages 30--39. JMLR, Inc., 2018.

\bibitem[ASZ20a]{hz}
Jayadev Acharya, Ziteng Sun, and Huanyu Zhang.
\newblock Personal communication, 2020.

\bibitem[ASZ20b]{AcharyaSZ20}
Jayadev Acharya, Ziteng Sun, and Huanyu Zhang.
\newblock Differentially private assouad, fano, and le cam.
\newblock {\em arXiv preprint arXiv:2004.06830}, 2020.

\bibitem[BBC{\etalchar{+}}19]{BezakovaBCSV19}
Ivona Bezakova, Antonio Blanca, Zongchen Chen, Daniel {\v{S}}tefankovi{\v{c}},
  and Eric Vigoda.
\newblock Lower bounds for testing graphical models: Colorings and
  antiferromagnetic {I}sing models.
\newblock In {\em Proceedings of the 32nd Annual Conference on Learning
  Theory}, COLT '19, pages 283--298, 2019.

\bibitem[BBKN14]{BeimelBKN14}
Amos Beimel, Hai Brenner, Shiva~Prasad Kasiviswanathan, and Kobbi Nissim.
\newblock Bounds on the sample complexity for private learning and private data
  release.
\newblock {\em Machine Learning}, 94(3):401--437, 2014.

\bibitem[BDMN05]{BlumDMN05}
Avrim Blum, Cynthia Dwork, Frank McSherry, and Kobbi Nissim.
\newblock Practical privacy: The {SuLQ} framework.
\newblock In {\em Proceedings of the 24th ACM SIGMOD-SIGACT-SIGART Symposium on
  Principles of Database Systems}, PODS '05, pages 128--138, New York, NY, USA,
  2005. ACM.

\bibitem[BGS14]{BreslerGS14b}
Guy Bresler, David Gamarnik, and Devavrat Shah.
\newblock Structure learning of antiferromagnetic {I}sing models.
\newblock In {\em Advances in Neural Information Processing Systems 27}, NIPS
  '14, pages 2852--2860. Curran Associates, Inc., 2014.

\bibitem[Bha19]{Bhattacharya19}
Bhaswar~B. Bhattacharya.
\newblock A general asymptotic framework for distribution-free graph-based
  two-sample tests.
\newblock {\em Journal of the Royal Statistical Society: Series B (Statistical
  Methodology)}, 81(3):575--602, 2019.

\bibitem[BKSW19]{BunKSW19}
Mark Bun, Gautam Kamath, Thomas Steinke, and Zhiwei~Steven Wu.
\newblock Private hypothesis selection.
\newblock In {\em Advances in Neural Information Processing Systems 32},
  NeurIPS '19. Curran Associates, Inc., 2019.

\bibitem[BM16]{BhattacharyaM16}
Bhaswar~B. Bhattacharya and Sumit Mukherjee.
\newblock Inference in {I}sing models.
\newblock {\em Bernoulli}, 2016.

\bibitem[BMS{\etalchar{+}}17]{BernsteinMSSHM17}
Garrett Bernstein, Ryan McKenna, Tao Sun, Daniel Sheldon, Michael Hay, and
  Gerome Miklau.
\newblock Differentially private learning of undirected graphical models using
  collective graphical models.
\newblock In {\em Proceedings of the 34th International Conference on Machine
  Learning}, ICML '17, pages 478--487. JMLR, Inc., 2017.

\bibitem[BNSV15]{BunNSV15}
Mark Bun, Kobbi Nissim, Uri Stemmer, and Salil Vadhan.
\newblock Differentially private release and learning of threshold functions.
\newblock In {\em Proceedings of the 56th Annual IEEE Symposium on Foundations
  of Computer Science}, FOCS '15, pages 634--649, Washington, DC, USA, 2015.
  IEEE Computer Society.

\bibitem[Bre15]{Bresler15}
Guy Bresler.
\newblock Efficiently learning {I}sing models on arbitrary graphs.
\newblock In {\em Proceedings of the 47th Annual ACM Symposium on the Theory of
  Computing}, STOC '15, pages 771--782, New York, NY, USA, 2015. ACM.

\bibitem[BS16]{BunS16}
Mark Bun and Thomas Steinke.
\newblock Concentrated differential privacy: Simplifications, extensions, and
  lower bounds.
\newblock In {\em Proceedings of the 14th Conference on Theory of
  Cryptography}, TCC '16-B, pages 635--658, Berlin, Heidelberg, 2016. Springer.

\bibitem[BUV14]{BunUV14}
Mark Bun, Jonathan Ullman, and Salil Vadhan.
\newblock Fingerprinting codes and the price of approximate differential
  privacy.
\newblock In {\em Proceedings of the 46th Annual ACM Symposium on the Theory of
  Computing}, STOC '14, pages 1--10, New York, NY, USA, 2014. ACM.

\bibitem[Cha05]{Chatterjee05}
Sourav Chatterjee.
\newblock {\em Concentration Inequalities with Exchangeable Pairs}.
\newblock PhD thesis, Stanford University, June 2005.

\bibitem[CL68]{ChowL68}
C.K. Chow and C.N. Liu.
\newblock Approximating discrete probability distributions with dependence
  trees.
\newblock {\em IEEE Transactions on Information Theory}, 14(3):462--467, 1968.

\bibitem[Cla10]{Clarkson10}
Kenneth~L Clarkson.
\newblock Coresets, sparse greedy approximation, and the frank-wolfe algorithm.
\newblock {\em ACM Transactions on Algorithms (TALG)}, 6(4):63, 2010.

\bibitem[CRJ19]{ChowdhuryRJ19}
Amrita~Roy Chowdhury, Theodoros Rekatsinas, and Somesh Jha.
\newblock Data-dependent differentially private parameter learning for directed
  graphical models.
\newblock {\em arXiv preprint arXiv:1905.12813}, 2019.

\bibitem[CT06]{CsiszarT06}
Imre Csisz{\'a}r and Zsolt Talata.
\newblock Consistent estimation of the basic neighborhood of {M}arkov random
  fields.
\newblock {\em The Annals of Statistics}, 34(1):123--145, 2006.

\bibitem[CWZ19]{CaiWZ19}
T.~Tony Cai, Yichen Wang, and Linjun Zhang.
\newblock The cost of privacy: Optimal rates of convergence for parameter
  estimation with differential privacy.
\newblock {\em arXiv preprint arXiv:1902.04495}, 2019.

\bibitem[DDK17]{DaskalakisDK17}
Constantinos Daskalakis, Nishanth Dikkala, and Gautam Kamath.
\newblock Concentration of multilinear functions of the {I}sing model with
  applications to network data.
\newblock In {\em Advances in Neural Information Processing Systems 30}, NIPS
  '17. Curran Associates, Inc., 2017.

\bibitem[DDK18]{DaskalakisDK18}
Constantinos Daskalakis, Nishanth Dikkala, and Gautam Kamath.
\newblock Testing {I}sing models.
\newblock In {\em Proceedings of the 29th Annual ACM-SIAM Symposium on Discrete
  Algorithms}, SODA '18, pages 1989--2007, Philadelphia, PA, USA, 2018. SIAM.

\bibitem[DDK19]{DaskalakisDK19}
Constantinos Daskalakis, Nishanth Dikkala, and Gautam Kamath.
\newblock Testing {I}sing models.
\newblock {\em IEEE Transactions on Information Theory}, 65(11):6829--6852,
  2019.

\bibitem[dH12]{coupling}
Frank den Hollander.
\newblock Probability theory: The coupling method.
\newblock {\em Lecture notes available online (http://websites. math.
  leidenuniv. nl/probability/lecturenotes/CouplingLectures. pdf)}, 2012.

\bibitem[DHS15]{DiakonikolasHS15}
Ilias Diakonikolas, Moritz Hardt, and Ludwig Schmidt.
\newblock Differentially private learning of structured discrete distributions.
\newblock In {\em Advances in Neural Information Processing Systems 28}, NIPS
  '15, pages 2566--2574. Curran Associates, Inc., 2015.

\bibitem[{Dif}17]{AppleDP17}
{Differential Privacy Team, Apple}.
\newblock Learning with privacy at scale.
\newblock
  \url{https://machinelearning.apple.com/docs/learning-with-privacy-at-scale/appledifferentialprivacysystem.pdf},
  December 2017.

\bibitem[DKY17]{DingKY17}
Bolin Ding, Janardhan Kulkarni, and Sergey Yekhanin.
\newblock Collecting telemetry data privately.
\newblock In {\em Advances in Neural Information Processing Systems 30}, NIPS
  '17, pages 3571--3580. Curran Associates, Inc., 2017.

\bibitem[DL09]{DworkL09}
Cynthia Dwork and Jing Lei.
\newblock Differential privacy and robust statistics.
\newblock In {\em Proceedings of the 41st Annual ACM Symposium on the Theory of
  Computing}, STOC '09, pages 371--380, New York, NY, USA, 2009. ACM.

\bibitem[DLS{\etalchar{+}}17]{DajaniLSKRMGDGKKLSSVA17}
Aref~N. Dajani, Amy~D. Lauger, Phyllis~E. Singer, Daniel Kifer, Jerome~P.
  Reiter, Ashwin Machanavajjhala, Simson~L. Garfinkel, Scot~A. Dahl, Matthew
  Graham, Vishesh Karwa, Hang Kim, Philip Lelerc, Ian~M. Schmutte, William~N.
  Sexton, Lars Vilhuber, and John~M. Abowd.
\newblock The modernization of statistical disclosure limitation at the {U.S.}
  census bureau, 2017.
\newblock Presented at the September 2017 meeting of the Census Scientific
  Advisory Committee.

\bibitem[DMNS06]{DworkMNS06}
Cynthia Dwork, Frank McSherry, Kobbi Nissim, and Adam Smith.
\newblock Calibrating noise to sensitivity in private data analysis.
\newblock In {\em Proceedings of the 3rd Conference on Theory of Cryptography},
  TCC '06, pages 265--284, Berlin, Heidelberg, 2006. Springer.

\bibitem[DMR11]{DaskalakisMR11}
Constantinos Daskalakis, Elchanan Mossel, and S{\'e}bastien Roch.
\newblock Evolutionary trees and the {I}sing model on the {B}ethe lattice: A
  proof of {S}teel's conjecture.
\newblock {\em Probability Theory and Related Fields}, 149(1):149--189, 2011.

\bibitem[DMR18]{DevroyeMR18a}
Luc Devroye, Abbas Mehrabian, and Tommy Reddad.
\newblock The minimax learning rate of normal and {I}sing undirected graphical
  models.
\newblock {\em arXiv preprint arXiv:1806.06887}, 2018.

\bibitem[DR16]{DworkR16}
Cynthia Dwork and Guy~N. Rothblum.
\newblock Concentrated differential privacy.
\newblock {\em arXiv preprint arXiv:1603.01887}, 2016.

\bibitem[DSS{\etalchar{+}}15]{DworkSSUV15}
Cynthia Dwork, Adam Smith, Thomas Steinke, Jonathan Ullman, and Salil Vadhan.
\newblock Robust traceability from trace amounts.
\newblock In {\em Proceedings of the 56th Annual IEEE Symposium on Foundations
  of Computer Science}, FOCS '15, pages 650--669, Washington, DC, USA, 2015.
  IEEE Computer Society.

\bibitem[Ell93]{Ellison93}
Glenn Ellison.
\newblock Learning, local interaction, and coordination.
\newblock {\em Econometrica}, 61(5):1047--1071, 1993.

\bibitem[EPK14]{ErlingssonPK14}
{\'U}lfar Erlingsson, Vasyl Pihur, and Aleksandra Korolova.
\newblock {RAPPOR}: Randomized aggregatable privacy-preserving ordinal
  response.
\newblock In {\em Proceedings of the 2014 ACM Conference on Computer and
  Communications Security}, CCS '14, pages 1054--1067, New York, NY, USA, 2014.
  ACM.

\bibitem[Fel04]{Felsenstein04}
Joseph Felsenstein.
\newblock {\em Inferring Phylogenies}.
\newblock Sinauer Associates Sunderland, 2004.

\bibitem[FLNP00]{FriedmanLNP00}
Nir Friedman, Michal Linial, Iftach Nachman, and Dana Pe'er.
\newblock Using {B}ayesian networks to analyze expression data.
\newblock {\em Journal of Computational Biology}, 7(3-4):601--620, 2000.

\bibitem[GAH{\etalchar{+}}14]{GaboardiAHRW14}
Marco Gaboardi, Emilio Jes{\'{u}}s~Gallego Arias, Justin Hsu, Aaron Roth, and
  Zhiwei~Steven Wu.
\newblock Dual query: Practical private query release for high dimensional
  data.
\newblock In {\em Proceedings of the 31th International Conference on Machine
  Learning, {ICML} 2014, Beijing, China, 21-26 June 2014}, pages 1170--1178,
  2014.

\bibitem[GG86]{GemanG86}
Stuart Geman and Christine Graffigne.
\newblock {M}arkov random field image models and their applications to computer
  vision.
\newblock In {\em Proceedings of the International Congress of Mathematicians},
  pages 1496--1517. American Mathematical Society, 1986.

\bibitem[GLP18]{GheissariLP18}
Reza Gheissari, Eyal Lubetzky, and Yuval Peres.
\newblock Concentration inequalities for polynomials of contracting {I}sing
  models.
\newblock {\em Electronic Communications in Probability}, 23(76):1--12, 2018.

\bibitem[HKM17]{HamiltonKM17}
Linus Hamilton, Frederic Koehler, and Ankur Moitra.
\newblock Information theoretic properties of {M}arkov random fields, and their
  algorithmic applications.
\newblock In {\em Advances in Neural Information Processing Systems 30}, NIPS
  '17. Curran Associates, Inc., 2017.

\bibitem[HR10]{HardtR10}
Moritz Hardt and Guy~N. Rothblum.
\newblock A multiplicative weights mechanism for privacy-preserving data
  analysis.
\newblock In {\em Proceedings of the 51st Annual IEEE Symposium on Foundations
  of Computer Science}, FOCS '10, pages 61--70, Washington, DC, USA, 2010. IEEE
  Computer Society.

\bibitem[HT10]{HardtT10}
Moritz Hardt and Kunal Talwar.
\newblock On the geometry of differential privacy.
\newblock In {\em Proceedings of the 42nd Annual ACM Symposium on the Theory of
  Computing}, STOC '10, pages 705--714, New York, NY, USA, 2010. ACM.

\bibitem[Isi25]{Ising25}
Ernst Ising.
\newblock Beitrag zur theorie des ferromagnetismus.
\newblock {\em Zeitschrift f{\"u}r Physik A Hadrons and Nuclei},
  31(1):253--258, 1925.

\bibitem[Jag13]{Jaggi13}
Martin Jaggi.
\newblock Revisiting frank-wolfe: Projection-free sparse convex optimization.
\newblock In {\em ICML (1)}, pages 427--435, 2013.

\bibitem[JJR11]{JalaliJR11}
Ali Jalali, Christopher~C. Johnson, and Pradeep~K. Ravikumar.
\newblock On learning discrete graphical models using greedy methods.
\newblock In {\em Advances in Neural Information Processing Systems 24}, NIPS
  '11, pages 1935--1943. Curran Associates, Inc., 2011.

\bibitem[JRVS11]{JalaliRVS11}
Ali Jalali, Pradeep~K. Ravikumar, Vishvas Vasuki, and Sujay Sanghavi.
\newblock On learning discrete graphical models using group-sparse
  regularization.
\newblock In {\em Proceedings of the 14th International Conference on
  Artificial Intelligence and Statistics}, AISTATS '11, pages 378--387. JMLR,
  Inc., 2011.

\bibitem[KKMN09]{KorolovaKMN09}
Aleksandra Korolova, Krishnaram Kenthapadi, Nina Mishra, and Alexandros
  Ntoulas.
\newblock Releasing search queries and clicks privately.
\newblock In {\em Proceedings of the 18th International World Wide Web
  Conference}, WWW '09, pages 171--180, New York, NY, USA, 2009. ACM.

\bibitem[KLSU19]{KamathLSU19}
Gautam Kamath, Jerry Li, Vikrant Singhal, and Jonathan Ullman.
\newblock Privately learning high-dimensional distributions.
\newblock In {\em Proceedings of the 32nd Annual Conference on Learning
  Theory}, COLT '19, pages 1853--1902, 2019.

\bibitem[KM17]{KlivansM17}
Adam Klivans and Raghu Meka.
\newblock Learning graphical models using multiplicative weights.
\newblock In {\em Proceedings of the 58th Annual IEEE Symposium on Foundations
  of Computer Science}, FOCS '17, pages 343--354, Washington, DC, USA, 2017.
  IEEE Computer Society.

\bibitem[KU20]{KamathU20}
Gautam Kamath and Jonathan Ullman.
\newblock A primer on private statistics.
\newblock {\em arXiv preprint arXiv:2005.00010}, 2020.

\bibitem[KV18]{KarwaV18}
Vishesh Karwa and Salil Vadhan.
\newblock Finite sample differentially private confidence intervals.
\newblock In {\em Proceedings of the 9th Conference on Innovations in
  Theoretical Computer Science}, ITCS '18, pages 44:1--44:9, Dagstuhl, Germany,
  2018. Schloss Dagstuhl--Leibniz-Zentrum fuer Informatik.

\bibitem[LAFH01]{LagorAFH01}
Charles Lagor, Dominik Aronsky, Marcelo Fiszman, and Peter~J. Haug.
\newblock Automatic identification of patients eligible for a pneumonia
  guideline: comparing the diagnostic accuracy of two decision support models.
\newblock {\em Studies in Health Technology and Informatics}, 84(1):493--497,
  2001.

\bibitem[LPW09]{LevinPW09}
David~A. Levin, Yuval Peres, and Elizabeth~L. Wilmer.
\newblock {\em {M}arkov Chains and Mixing Times}.
\newblock American Mathematical Society, 2009.

\bibitem[LVMC18]{LokhovVMC18}
Andrey~Y. Lokhov, Marc Vuffray, Sidhant Misra, and Michael Chertkov.
\newblock Optimal structure and parameter learning of {I}sing models.
\newblock {\em Science Advances}, 4(3):e1700791, 2018.

\bibitem[MdCCU16]{MartindelCampoCU16}
Abraham Mart{\'\i}n~del Campo, Sarah Cepeda, and Caroline Uhler.
\newblock Exact goodness-of-fit testing for the {I}sing model.
\newblock {\em Scandinavian Journal of Statistics}, 2016.

\bibitem[MMY18]{MukherjeeMY18}
Rajarshi Mukherjee, Sumit Mukherjee, and Ming Yuan.
\newblock Global testing against sparse alternatives under {I}sing models.
\newblock {\em The Annals of Statistics}, 46(5):2062--2093, 2018.

\bibitem[MS10]{MontanariS10}
Andrea Montanari and Amin Saberi.
\newblock The spread of innovations in social networks.
\newblock {\em Proceedings of the National Academy of Sciences},
  107(47):20196--20201, 2010.

\bibitem[MSM19]{McKennaSM19}
Ryan McKenna, Daniel Sheldon, and Gerome Miklau.
\newblock Graphical-model based estimation and inference for differential
  privacy.
\newblock {\em arXiv preprint arXiv:1901.09136}, 2019.

\bibitem[NRS07]{NissimRS07}
Kobbi Nissim, Sofya Raskhodnikova, and Adam Smith.
\newblock Smooth sensitivity and sampling in private data analysis.
\newblock In {\em Proceedings of the 39th Annual ACM Symposium on the Theory of
  Computing}, STOC '07, pages 75--84, New York, NY, USA, 2007. ACM.

\bibitem[RH17]{RigolletH17}
Philippe Rigollet and Jan-Christian H\"utter.
\newblock High dimensional statistics.
\newblock \url{http://www-math.mit.edu/~rigollet/PDFs/RigNotes17.pdf}, 2017.
\newblock Lecture notes.

\bibitem[RWL10]{RavikumarWL10}
Pradeep Ravikumar, Martin~J. Wainwright, and John~D. Lafferty.
\newblock High-dimensional {I}sing model selection using $\ell_1$-regularized
  logistic regression.
\newblock {\em The Annals of Statistics}, 38(3):1287--1319, 2010.

\bibitem[Smi11]{Smith11}
Adam Smith.
\newblock Privacy-preserving statistical estimation with optimal convergence
  rates.
\newblock In {\em Proceedings of the 43rd Annual ACM Symposium on the Theory of
  Computing}, STOC '11, pages 813--822, New York, NY, USA, 2011. ACM.

\bibitem[SU17]{SteinkeU17a}
Thomas Steinke and Jonathan Ullman.
\newblock Between pure and approximate differential privacy.
\newblock {\em The Journal of Privacy and Confidentiality}, 7(2):3--22, 2017.

\bibitem[SW12]{SanthanamW12}
Narayana~P. Santhanam and Martin~J. Wainwright.
\newblock Information-theoretic limits of selecting binary graphical models in
  high dimensions.
\newblock {\em IEEE Transactions on Information Theory}, 58(7):4117--4134,
  2012.

\bibitem[TTZ14]{TalwarTZ14}
Kunal Talwar, Abhradeep Thakurta, and Li~Zhang.
\newblock Private empirical risk minimization beyond the worst case: The effect
  of the constraint set geometry.
\newblock {\em arXiv preprint arXiv:1411.5417}, 2014.

\bibitem[TTZ15]{TalwarTZ15}
Kunal Talwar, Abhradeep Thakurta, and Li~Zhang.
\newblock Nearly-optimal private {LASSO}.
\newblock In {\em Advances in Neural Information Processing Systems 28}, NIPS
  '15, pages 3025--3033. Curran Associates, Inc., 2015.

\bibitem[Vad17]{Vadhan17}
Salil Vadhan.
\newblock The complexity of differential privacy.
\newblock In Yehuda Lindell, editor, {\em Tutorials on the Foundations of
  Cryptography: Dedicated to Oded Goldreich}, chapter~7, pages 347--450.
  Springer International Publishing AG, Cham, Switzerland, 2017.

\bibitem[VMLC16]{VuffrayMLC16}
Marc Vuffray, Sidhant Misra, Andrey Lokhov, and Michael Chertkov.
\newblock Interaction screening: Efficient and sample-optimal learning of
  {I}sing models.
\newblock In {\em Advances in Neural Information Processing Systems 29}, NIPS
  '16, pages 2595--2603. Curran Associates, Inc., 2016.

\bibitem[VTB{\etalchar{+}}19]{neworacle}
Giuseppe Vietri, Grace Tian, Mark Bun, Thomas Steinke, and Zhiwei~Steven Wu.
\newblock New oracle efficient algorithms for private synthetic data release.
\newblock {\em NeurIPS PriML workshop}, 2019.

\bibitem[WSD19]{WuSD19}
Shanshan Wu, Sujay Sanghavi, and Alexandros~G. Dimakis.
\newblock Sparse logistic regression learns all discrete pairwise graphical
  models.
\newblock In {\em Advances in Neural Information Processing Systems 32},
  NeurIPS '19, pages 8069--8079. Curran Associates, Inc., 2019.

\bibitem[ZLA03]{ZLA03}
Bianca Zadrozny, John Langford, and Naoki Abe.
\newblock Cost-sensitive learning by cost-proportionate example weighting.
\newblock In {\em Proceedings of the 3rd {IEEE} International Conference on
  Data Mining {(ICDM} 2003), 19-22 December 2003, Melbourne, Florida, {USA}},
  page 435, 2003.

\end{thebibliography}

\end{document}